\documentclass[11pt]{amsart}
\usepackage{txfonts}
\usepackage{amsmath}
\usepackage{mathrsfs}
\usepackage{} \textwidth=14.5cm \oddsidemargin=1cm
\usepackage{yhmath}
\usepackage{amsxtra}
\usepackage{amscd}
\usepackage{amsthm}
\usepackage{amsfonts}
\usepackage{amssymb}
\usepackage{eucal}
\usepackage{pmgraph}
\usepackage{stmaryrd}
\usepackage[usenames,dvipsnames]{color}
\usepackage{xcolor}
\usepackage{bbding}
\usepackage{pifont}

\input prepictex
\input pictex
\input postpictex

\newtheorem{thm}{Theorem}[section]
\newtheorem{cor}[thm]{Corollary}
\newtheorem{lem}[thm]{Lemma}
\newtheorem{prop}[thm]{Proposition}

\theoremstyle{definition}

\theoremstyle{remark}
\newtheorem{rem}[thm]{Remark}

\numberwithin{equation}{section}

\begin{document}

\newcommand{\thmref}[1]{Theorem~\ref{#1}}
\newcommand{\secref}[1]{Section~\ref{#1}}
\newcommand{\lemref}[1]{Lemma~\ref{#1}}
\newcommand{\propref}[1]{Proposition~\ref{#1}}
\newcommand{\corref}[1]{Corollary~\ref{#1}}
\newcommand{\remref}[1]{Remark~\ref{#1}}
\newcommand{\eqnref}[1]{(\ref{#1})}
\newcommand{\exref}[1]{Example~\ref{#1}}

\newcommand{\nc}{\newcommand}
\nc{\Z}{{\mathbb Z}}
\nc{\C}{{\mathbb C}}
\nc{\N}{{\mathbb N}}
\nc{\F}{{\mf F}}
\nc{\Q}{\ol{Q}}
\nc{\la}{\lambda}
\nc{\ep}{\epsilon}
\nc{\h}{\mathfrak h}
\nc{\n}{\mf n}
\nc{\G}{{\mathfrak g}}
\nc{\DG}{\widetilde{\mathfrak g}}
\nc{\SG}{\overline{\mathfrak g}}
\nc{\D}{\mc D} \nc{\Li}{{\mc L}} \nc{\La}{\Lambda} \nc{\is}{{\mathbf
i}} \nc{\V}{\mf V} \nc{\bi}{\bibitem} \nc{\NS}{\mf N}
\nc{\dt}{\mathord{\hbox{${\frac{d}{d t}}$}}} \nc{\E}{\mc E}
\nc{\ba}{\tilde{\pa}} \nc{\half}{\frac{1}{2}} \nc{\mc}{\mathcal}
\nc{\mf}{\mathfrak} \nc{\hf}{\frac{1}{2}}
\nc{\hgl}{\widehat{\mathfrak{gl}}} \nc{\gl}{{\mathfrak{gl}}}
\nc{\hz}{\hf+\Z}
\nc{\dinfty}{{\infty\vert\infty}} \nc{\SLa}{\overline{\Lambda}}
\nc{\SF}{\overline{\mathfrak F}} \nc{\SP}{\overline{\mathcal P}}
\nc{\U}{\mathfrak u} \nc{\SU}{\overline{\mathfrak u}}
\nc{\ov}{\overline}
\nc{\wt}{\widetilde}
\nc{\osp}{\mf{osp}}
\nc{\spo}{\mf{spo}}
\nc{\hosp}{\widehat{\mf{osp}}}
\nc{\hspo}{\widehat{\mf{spo}}}
\nc{\I}{\mathbb{I}}
\nc{\X}{\mathbb{X}}
\nc{\Y}{\mathbb{Y}}
\nc{\hh}{\widehat{\mf{h}}}
\nc{\cc}{{\mathfrak c}}
\nc{\dd}{{\mathfrak d}}
\nc{\aaa}{{\mf A}}
\nc{\xx}{{\mf x}}
\nc{\wty}{\widetilde{\mathbb Y}}
\nc{\ovy}{\overline{\mathbb Y}}
\nc{\vep}{\bar{\epsilon}}
\nc{\w}{\widetilde}
\nc{\btd}{\blacktriangledown}

%%%%%%%%%%%%%%%%%%%%%%%%%%%
\nc{\mG}{{\mathcal{G}}}
\nc{\mGb}{\overline{{\mathcal{G}}}}
\nc{\mGt}{\widetilde{\mathcal{G}}}
\nc{\mGx}{\mathcal{G}^\xx}
\nc{\mGbx}{\overline{\mathcal{G}^\xx}}
\nc{\mGtx}{\widetilde{\mathcal{G}^\xx}}

%%%%%%%%%%%%%%%%%%%%%%%%%%%
\nc{\OO}{\mathcal {O}}
\nc{\OOb}{\overline{\mathcal {O}}}
\nc{\OOt}{\widetilde{\mathcal {O}}}
\nc{\OOn}{{\mathcal {O}_n}}
\nc{\OObn}{{\overline{\mathcal {O}}_n}}
\nc{\OOtn}{{\widetilde{\mathcal {O}}_n}}
\nc{\OOx}{{\mathcal {O}^\xx}}
\nc{\OObx}{{\overline{\mathcal {O}}^\xx}}
\nc{\OOtx}{{\widetilde{\mathcal {O}}^\xx}}
\nc{\OOxn}{{\mathcal {O}^\xx_n}}
\nc{\OObxn}{{\overline{\mathcal {O}}^\xx_n}}
\nc{\OOtxn}{{\widetilde{\mathcal {O}}^\xx_n}}

\advance\headheight by 2pt

\title[Solutions of super KZ equations]
{Solutions of super Knizhnik-Zamolodchikov equations}

\author[Bintao Cao]
{Bintao Cao}
%\thanks{ {1.} On behalf of all authors, the
%corresponding author states that there is no conflict of interest.}
\address{School of Mathematics and Statistics, Yunnan University, Kunming, China, 650500; School of Mathematics, Sun
Yat-sen University, Guangzhou, China, 510275}
\email{btcao@ynu.edu.cn}

\author[Ngau Lam]{Ngau Lam}
\address{Department of Mathematics, National Cheng Kung University, Tainan, Taiwan 70101}
\email{nlam@mail.ncku.edu.tw}

\begin{abstract}
We  establish an explicit bijection between the sets of singular solutions of the (super) KZ equations associated to the  Lie superalgebra, of infinite rank, of type $\mf{a, b,c,d}$ and to the corresponding Lie algebra.
As a consequence,  the singular solutions of the super KZ equations associated to the  classical Lie superalgebra, of finite rank, of type $\mf{a, b,c,d}$ for the tensor product of certain parabolic Verma modules (resp., irreducible modules) are obtained from the singular solutions of the KZ equations for the tensor product of the corresponding parabolic Verma modules (resp., irreducible modules) over the corresponding Lie algebra of sufficiently large rank, and vice versa. The analogous results for some special kinds of  trigonometric (super) KZ equations are obtained.
\end{abstract}

 \maketitle

  \setcounter{tocdepth}{1}

\section{Introduction}\label{introduction}

The Knizhnik-Zamolodchikov equations (KZ equations) were first discovered by V. G. Knizhnik and A. B. Zamolodchikov \cite{KZ} from studying the Wess-Zumino-Novikov-Witten model in $2$-dimensional conformal field theory.
A mathematical treatment,  in the case $\mf{sl}_2$,  was first given in \cite{TK} (see also \cite{EFK}).
These equations compose a system of compatible complex partial differential equations, which can be derived using either the formalism of Lie algebras or that of vertex algebras.
Let us explain with more details for our settings. Let  $\G$ be a  finite-dimensional simple complex Lie algebra and let $(\ ,\ )$ be a non-degenerate invariant bilinear form  on $\G$.  Let  $\{I_a\, |\, a\in \mathbb{I}\}$ be a basis of $\G$ and $\{I^a\, |\, a\in \mathbb{I}\}$ the dual basis with respect to the bilinear form $(\ ,\ )$.
Let $V_1, \ldots, V_\ell$ be highest weight $\G$-modules and $V=V_1\otimes\cdots\otimes V_\ell$. Let $\Omega:=\sum_a I_a\otimes I^a$ denote the Casimir symmetric tensor  and let $\Omega^{(ij)}:=\sum_a I_a^{(i)}\otimes I^a{}^{(j)}$, where $x^{(i)}$ means the action of $x$ on the $i$-th component of $V$, for $x\in\G$ and $i=1, \ldots, \ell$. The KZ equations with values in $V$
constitute a system of differential equations
\begin{equation}\label{Intro: KZ}
\kappa\frac{\partial \psi}{\partial z_i}=\sum_{j=1\atop j\neq i}^\ell\frac{\Omega^{(ij)}}{z_i-z_j}\psi, \qquad \text{ for $i=1,\ldots, \ell$,}
\end{equation}
with a fixed parameter $\kappa \in \C$, where $\psi$ is a $V$-valued function on a nonempty open subset $U$ in the configuration space ${\bf X}_\ell:=\{(z_1, \ldots,z_\ell)\in \C^\ell\,|\, z_i\not=z_j, \,\, \hbox{for any $i\not=j$}\}$ of $\ell$ distinct points in $\C$.
The solutions were found in \cite{CF}, \cite{TK},  \cite{M}, \cite{DJMM}, \cite{FR}, et cetera, and the general case was obtained by V.V. Schechtman and A. N. Varchenko \cite{SV2} (see \cite{EFK} and \cite{V} for more information).
The super KZ equations compose a system of partial differential equations formulated analogously to the equations \eqnref{Intro: KZ} for Lie superalgebras (see, for examples, \cite{GZZ, KM}). Similar to the non-super cases, the super KZ equations also play important roles for the representation theory for quantum supergroups. (See, for example, \cite{Ge}).  However, the analogous results for the solutions of the super KZ equations are very rare compared with the non-super cases.

The super duality conjecture was originally developed in \cite{CWZ} (see also \cite{CZ}) and in full generality in \cite{CW1} for finding characters of certain irreducible modules over the general linear superalgebras  $\mf {gl}(m|n)$. The super duality conjecture was proved in \cite{CL} for type $\mf{a}$ and  in \cite{CLW} for type $\mf{b,c,d}$ showing that there is an equivalence of tensor categories between the parabolic BGG category $\overline{\mc{O}}$ for Lie superalgebra, of infinite rank, of type $\mf{a,b,c}$ or $\mf d$ and the parabolic BGG category $\mc{O}$ for the corresponding Lie algebra (see \propref{thm:equivalence} below for precise statements).

In this article, we apply the results of the super duality to connect the singular solutions (that is, the solution $\psi$ takes values in the space spanned by the singular vectors in $V$) of super KZ equations associated to the Lie superalgebra, of infinite rank, of type $\mf{a,b,c}$ or $\mf d$ and of  KZ equations associated to the corresponding Lie algebra.
Formulating the KZ equations associated to the Lie algebras of infinite rank and the super KZ equations associated to the Lie superalgebra of infinite rank is an essential prerequisite.
For each positive integer $n$, let $\G_n$ and $\SG_n$  denote the Lie algebra and Lie superalgebra of type $\mf{a,b,c}$ or $\mf d$ in the settings of the super duality, respectively (see Subsection~\ref{Central Ext} below). There are inclusions of Lie superalgebras $\SG_n\subset \SG_{n+1}$ and of Lie algebras $\G_n\subset \G_{n+1}$,  for $n\in\N$. Let $\SG :=\SG_\infty :=\cup_{n=1}^\infty \SG_n$ and $\G :=\G_\infty :=\cup_{n=1}^\infty \G_n$. The Lie algebra $\G$ is the counterpart of the Lie superalgebra $\SG$ for the super duality.
The KZ equations associated to the  Lie algebra $\G$ and the super KZ equations associated to the  Lie superalgebra $\SG$ are defined in Subsection~\ref{SKZ} and an explicit bijection between the sets of singular solutions of the (super) KZ equations associated to the  Lie algebra $\G$ and to the Lie superalgebra $\SG$ is obtained in Subsection~\ref{Sec: KZ_infty}.  Also we show that $\psi$ is also a solution/singular solution of the KZ equations (resp., super KZ equations) for  ${\G}_n$ (resp., $\SG_n$) with $n\ge k$ including $n=\infty$ if  $\psi$ is a solution/singular solution of the KZ equations (resp., super KZ equations) for  ${\G}_k$ (resp., $\SG_k$).
The Lie superalgebra $\SG_n$ (resp., $\G_n$) is a trivial central extension of $\ov{\mc G}_n$ (resp., ${\mc G}_n$), where $\ov{\mc G}_n$ (resp., ${\mc G}_n$) is either the general linear Lie superalgebra $\mf{gl}(m|n)$ (the general linear Lie algebra $\mf{gl}(m+n)$), Lie superalgebra of type  $\mf{osp}$ (resp., orthogonal Lie algebra) or Lie superalgebra of type  $\mf{spo}$ (resp.,  symplectic Lie algebra). See Subsections~\ref{Subsec:Lie superalgebras} and \ref{Central Ext} below. For $n\in \N$, we show in \propref{ncetoce} that the difference of the solutions/singular solutions of the KZ equations (resp., super KZ equations) for ${\mc G}_n$ (resp., $\ov{\mc G}_n$) and for ${\G}_n$ (resp., $\SG_n$) is by multiplying by a scalar function. Combing the results mentioned above, we have a bijection between the sets of singular solutions of the super KZ equations
for the tensor product of parabolic Verma modules (resp. irreducible modules) over $\ov{ \mc G}_n$
and the singular solutions of the KZ equations
 for the tensor product of the corresponding parabolic Verma modules (resp. irreducible modules) in $\mc G_k$
 for all sufficiently large $k$, and vice versa. As an application of our results, we give an explicit bijection between the sets of singular solutions of the KZ equations (non-super!) for the tensor product of infinite dimensional unitarizable modules of Lie algebra and the tensor product of integrable modules of corresponding Lie algebra since the super duality connects infinite dimensional unitarizable modules and integrable modules (see \cite[Section 2.3, 3.1]{CLW}, \cite{HLT, CLW2}).

The paper is organized as follows. We review the materials related to the super duality in \secref{sec:superalgebras}. The Lie superalgebras $\DG_n$, $\SG_n$ and $\G_n$ and their respective module categories $\wt{\mc{O}}_n$, $\ov{\mc{O}}_n$ and ${\mc{O}}_n$ are recalled.
In \secref{RKZ},  the Casimir operators and the Casimir symmetric tensors of the Lie superalgebras $\DG$, $\G$ and $\SG$ are defined and the results about the (super) KZ-equations are established. The analogous results of the (super) KZ equations for some special kinds of  trigonometric (super) KZ equations are obtained in Section~\ref{Sec:TKZ}.

\vskip 0.5cm
\noindent{\bf Acknowledgment.} The first author was partially supported by NSFC (Grant No.  11571374, 11521101 and 11771461).
A part of this research was done during the visit of the second author to Sun Yat-sen University.
The second author was partially supported by Ministry of Science and Technology grant 107-2115-M-006-005-MY2 of Taiwan and he thanks Sun Yat-sen University for hospitality and support.

\vskip 0.3cm
\noindent{\bf Notations.} Throughout the paper the symbols $\Z$, $\N$, and $\Z_+$ stand for the sets of all,
positive and non-negative integers, respectively. All vector spaces, algebras, tensor
products, et cetera, are over the field of complex numbers $\C$.

\bigskip

\section{Preliminaries}\label{sec:superalgebras}
 In this section, we first recall the infinite rank Lie (super)algebras $\G^\xx$, $\SG^\xx$ and
$\wt{\G}^\xx$ associated to the three Dynkin diagrams in \eqref{Dynkin:combined} below and their parabolic BGG categories $\OOx$ (resp., $\OObx$ and $\OOtx$) of $\G^\xx$- (resp., $\ov{\G}^\xx$- and $\wt{\G}^\xx$-)modules, where $\xx$ denotes one of the five types $\mf{a, b,b^\bullet,c,d}$. Then we recall the truncation functors on the categories $\OOx$, $\OObx$ and $\OOtx$ which relate the parabolic BGG categories $\OOxn$, $\OObxn$ and $\OOtxn$ of   finite-dimensional Lie subalgebras of $\G^\xx_n$, $\ov{\G}^\xx_n$ and $\wt{\G}^\xx_n$, respectively. We refer the readers to \cite[Sections 2 and 3]{CL} for type $\mf{a}$ and \cite[Sections 2 and 3]{CLW} for types $\mf{b,b^\bullet,c,d}$ for details (see also \cite[Sections 6.1 and 6.2]{CW2}).
Finally, we recall the tensor functors $T$ and $\ov T$ and their properties.
We fix $m\in\Z_+$ in this article.

For $m\in\Z_+$, we let $\w{\I}_m$ denote the following totally ordered set:
\begin{align*}
\cdots <\ov{\frac{3}{2}}
<\ov{1}<\ov{\hf}<\underbrace{\ov{-1}%<\ov{-2}
<\cdots<\ov{-m}}_m<\ov{0}<\underbrace{-m<\cdots<-1}_m
<\hf<1<\frac{3}{2}<\cdots
\end{align*}
and define the following subsets of $\w{\mathbb I}_m$:
\begin{align*}
{{\mathbb I}}_m &:=
\{\underbrace{\ov{-1},\ldots,\ov{-m}}_m,\ov{0},
\underbrace{-m,\ldots,-1}_m\}\cup \{\ov{1},\ov{2},\ov{3},\ldots\}
\cup\{1,2,3,\ldots\},\\
\ov{{\mathbb I}}_m &:=
\{\underbrace{\ov{-1},\ldots,\ov{-m}}_m,\ov{0},
\underbrace{-m,\ldots,-1}_m\}\cup
\{\ov{\frac{1}{2}},\ov{\frac{3}{2}},\ov{\frac{5}{2}},\ldots\}
\cup\{\hf,\frac{3}{2},\frac{5}{2},\ldots\},
 \\
\w{\I}^+_m &:=\{-m,\ldots,-1,\hf,1,\frac{3}{2},2,\ldots\}.
\end{align*}
We set $\{\ov{-1},\ldots,\ov{-m}\}\cup
\{-m,\ldots,-1\}=\emptyset$ for $m=0$.
For a subset $\X$ of $\wt{\I}_m$, let
$\X^\times:=\X\setminus\{\ov{0}\}$, $\X^+ :=\X\cap\wt{\I}_m^+$ and $\X(n)=\X\setminus\{j|j>n,j\in\frac{1}{2}\N\}$ for $n\in \N\cup\{\infty\}$.
For example,
$$
\wt{\I}^+_m(\infty)=\wt{\I}^+_m, \qquad {\I}^+_m(\infty)={\I}^+_m, \qquad
\ov{\I}^+_m(\infty)=\ov{\I}^+_m,
$$
\begin{align*}
\w{\mathbb I}^+_{m}(n)=\{\underbrace{-m,\ldots,-1}_m\}
\cup\{\frac{1}{2},1,\frac{3}{2},\ldots,n-\frac{1}{2},n\},\,\,\,
{{\mathbb I}}^+_{m}(n)={{\mathbb I}}^+_m\cap\w{\mathbb I}^+_{m}(n) \,\,\, \hbox{and} \,\,\, {\overline{\mathbb I}}^+_{m}(n)={\overline{\mathbb I}}^+_m\cap\w{\mathbb I}^+_{m}(n).
\end{align*}

\subsection{Lie superalgebras $\mc G^\xx$, $\ov{\mc G}^\xx$ and $\wt{\mc G}^\xx$}\label{Subsec:Lie superalgebras}
For a homogeneous element $v$ in a super vector space
$V=V_{\bar{0}}\oplus V_{\bar{1}}$ we denote by $|v|$ its
$\Z_2$-degree. For $m\in\Z_+$ and $n\in\N\cup\{\infty\}$, let $\wt{V}_{m}(n)$ denote the super space over $\C$ with
ordered basis $\{v_i\,|\,i\in\wt{\I}_m(n)\}$. We declare $|v_r|=|v_{\ov{r}}|=\bar{0}$, if
$r\in\Z\setminus\{0\}$, and $|v_r|=|v_{\ov{r}}|=\bar{1}$, if $r\in\hf+\Z_+$.  The parity
of the vector $v_{\ov{0}}$ is to be specified. With respect to this basis a linear map on
$\wt{V}_m(n)$ may be identified with a complex matrix $(a_{rs})_{r,s\in\wt{\mathbb{I}}_m(n)}$.
The Lie superalgebra $\gl(\wt{V}_m(n))$ is the Lie subalgebra of linear transformations on
$\wt{V}_m$ consisting of $(a_{rs})$ with $a_{rs}=0$ for all but finitely many $a_{rs}$'s.
Denote by $E_{rs}\in\gl(\wt{V}_m(n))$ the elementary matrix with $1$ at the $r$th row and
$s$th column and zero elsewhere.

The vector spaces $V_m(n)$ and $\ov{V}_m(n)$ are defined to be subspaces
of $\wt{V}_m(n)$ with ordered basis $\{v_i\}$ indexed by $\I_m(n)$ and
$\ov{\I}_m(n)$, respectively. The corresponding subspaces of $V_m(n)$,
$\ov{V}_m(n)$ and $\wt{V}_m(n)$ with basis vectors $v_i$, with $i$
indexed by $\I^\times_m(n)$, $\ov{\I}^\times_m(n)$ and
$\wt{\I}^\times_m(n)$, respectively, are denoted by $V^\times_m(n)$,
$\ov{V}^\times_m(n)$ and $\wt{V}^\times_m(n)$, respectively.
This gives
rise to Lie superalgebras $\gl(V_m(n))$, $\gl(\ov{V}_m(n))$,
$\gl(V^\times_m(n))$, $\gl(\ov{V}^\times_m(n))$ and
$\gl(\w{V}^\times_m(n))$.

\subsubsection{General linear superalgebras $\mGt^{\mf{a}}_n$}
\label{sec: A} For $m\in\Z_+$ and $n\in\N\cup\{\infty\}$, let $\mGt^{\mf{a}}_n$ denote the Lie subalgebra of $\gl(\wt{V}_m)$ spanned by $\{E_{i,j}\,|\, i,j\in \w{\I}^+_m(n)\}$ and let
$$
\mG^{\mf{a}}_n:=\mGt^{\mf{a}}_n\cap \gl({V}_m)\quad \hbox{and }\quad \mGb^{\mf{a}}_n:=\mGt^{\mf{a}}_n\cap \gl(\ov{V}_m).
$$
The sets $\{E_{i,j}\,|\, i,j\in \w{\I}^+_m(n)\}$,
$\{E_{i,j}\,|\, i,j\in {\I}^+_m(n)\}$ and $\{E_{i,j}\,|\, i,j\in \ov{\I}^+_m(n)\}$ are bases for
$\mGt^{\mf{a}}_n$, $\mG^{\mf{a}}_n$ and $\mGb^{\mf{a}}_n$, respectively.  The Dynkin diagrams for the Lie
(super)algebras $\mGt^{\mf{a}}_n$, $\mG^{\mf{a}}_n$ and $\mGb^{\mf{a}}_n$ are given in \eqref{Dynkin:combined} below with the corresponding notations.
The corresponding Cartan subalgebras have bases $\{E_i:=E_{i,i}\,|\, i\in \w{\I}^+_m(n)\}$,
$\{E_i\,|\, i\in {\I}^+_m(n)\}$ and $\{E_i\,|\, i\in \ov{\I}^+_m(n)\}$, respectively.

\subsubsection{Ortho-symplectic Lie superalgebras $\mGt^{\mf{b^\bullet}}_n$ and $\mGt^{\mf{c}}_n$}
\label{sec:skewsym}

In this subsection we set $|v_{\ov{0}}|=\bar{1}$. For $m\in\Z_+$ and $n\in\N\cup\{\infty\}$, we
define a non-degenerate skew-supersymmetric bilinear form
$(\cdot|\cdot)$ on $\wt{V}_m(n)$ by
\begin{eqnarray*}\label{symp:bilinear:form}
&&(v_{r}|v_{{s}}) =(v_{\ov{r}}|v_{\ov{s}})=0,\quad
(v_r|v_{\ov{s}}) =\delta_{rs}=-(-1)^{|v_r|\cdot|v_s|}(v_{\ov{s}}|v_r),
 \quad r,s\in\wt{\I}_m^+,\\
&&(v_{\ov{0}}|v_{\ov{0}})=1, \quad
(v_{\ov{0}}|v_{r})=(v_{\ov{0}}|v_{\ov{r}})=0,\quad
r\in\wt{\I}_m^+.\nonumber
\end{eqnarray*}
Restricting the form to $\wt{V}^\times_m(n)$, $V_m(n)$, $V^\times_m(n)$,
$\ov{V}_m(n)$ and $\ov{V}^\times_m(n)$ gives rise to non-degenerate
skew-supersymmetric bilinear forms that will again be denoted by
$(\cdot|\cdot)$.

Let $\mGt^{\mf{b^\bullet}}_n$, $\mG^{\mf{b^\bullet}}_n$, $\mGb^{\mf{b^\bullet}}_n$, $\mGt^{\mf{c}}_n$, $\mG^{\mf{c}}_n$ and $\mGb^{\mf{c}}_n$ denote the subalgebras of Lie superalgebras $\gl(\w{V}_m(n))$, $\gl({V}_m(n))$, $\gl(\ov{V}_m(n))$,
$\gl(\w{V}^\times_m(n))$, $\gl({V}^\times_m(n))$ and
$\gl(\ov{V}^\times_m(n))$ preserving the bilinear form $(\cdot|\cdot)$, respectively.
Their Dynkin diagrams are given in \eqref{Dynkin:combined} below with the corresponding notations.
The corresponding Cartan subalgebras have bases $\{E_r:=E_{rr}-E_{\ov{r},\ov{r}}\,|\, r\in \w{\I}^+_m(n)\}$,
$\{E_r\,|\, r\in {\I}^+_m(n)\}$ and  $\{E_r\,|\,r\in \ov{\I}^+_m(n)\}$, respectively.

\subsubsection{Ortho-symplectic Lie superalgebras $\mGt^{\mf{b}}_n$ and $\mGt^{\mf{d}}_n$}\label{sec:symm}

In this subsection, we set
$|v_{\ov{0}}|=\bar{0}$. For $m\in\Z_+$ and $n\in\N\cup\{\infty\}$, we define a non-degenerate supersymmetric bilinear form
$(\cdot|\cdot)$ on $\wt{V}_m(n)$ by
\begin{eqnarray*}\label{sym:bilinear:form}
&&(v_{r}|v_{{s}})=(v_{\ov{r}}|v_{\ov{s}})=0,\quad
(v_r|v_{\ov{s}})=\delta_{rs}=(-1)^{|v_r|\cdot|v_s|}(v_{\ov{s}}|v_r),
 \quad r,s\in\wt{\I}_m^+,\\
&&(v_{\ov{0}}|v_{\ov{0}})=1,\quad (v_{\ov{0}}|v_{r})
=(v_{\ov{0}}|v_{\ov{r}})=0,\quad r\in\wt{\I}_m^+.\nonumber
\end{eqnarray*}
Restricting the form to $\wt{V}^\times_m(n)$, $V_m(n)$, $V^\times_m(n)$, $\ov{V}_m(n)$ and
$\ov{V}^\times_m(n)$ gives respective non-degenerate supersymmetric bilinear forms that will
also be denoted by $(\cdot|\cdot)$.

Let $\mGt^{\mf{b}}_n$, $\mG^{\mf{b}}_n$, $\mGb^{\mf{b}}_n$, $\mGt^{\mf{d}}_n$, $\mG^{\mf{d}}_n$ and $\mGb^{\mf{d}}_n$ denote the subalgebras of Lie superalgebras $\gl(\w{V}_m(n))$, $\gl({V}_m(n))$, $\gl(\ov{V}_m(n))$,
$\gl(\w{V}^\times_m(n))$, $\gl({V}^\times_m(n))$ and
$\gl(\ov{V}^\times_m(n))$ preserving the bilinear form $(\cdot|\cdot)$, respectively.
Their Dynkin diagrams are given in \eqref{Dynkin:combined} below with the corresponding notations.
The corresponding Cartan subalgebras have bases $\{E_r:=E_{rr}-E_{\ov{r},\ov{r}}\,|\, r\in \w{\I}^+_m(n)\}$,
$\{E_r\,|\, r\in {\I}^+_m(n)\}$ and  $\{E_r\,|\,r\in \ov{\I}^+_m(n)\}$, respectively.

\subsection{Dynkin diagrams}
\label{dynkin}
Consider the free abelian group with basis
$\{\epsilon_{i},|i\in\w{\mathbb I}^+_m\}$, with a symmetric bilinear form $(\cdot,\cdot)$
given by
\begin{equation}\label{b.f. on root system}
(\epsilon_r,\epsilon_s)=(-1)^{2r}\delta_{rs}, \qquad r,s \in
\w{\mathbb I}^+_m.
\end{equation}
Let $|\epsilon_i|=0$ for $i\in\mathbb{Z}$ and $|\epsilon_j|=1$ for $j\in\frac{1}{2}+\mathbb{Z}$, for convenience.
We set
\begin{align*}\label{alpha:beta}
&\alpha_{\times}:=\epsilon_{-1}-\epsilon_{1/2},
 \quad\alpha_{j}
:=\epsilon_{j}-\epsilon_{j+1},\quad -m\le j\le -2,\\
&\beta_{\times}:=\epsilon_{-1}-\epsilon_{1},\quad
\alpha_{r}:=\epsilon_{r}-\epsilon_{r+1/2},\quad \beta_{r}
:=\epsilon_{r}-\epsilon_{r+1},\quad r\in\hf\N.\nonumber
\end{align*}

For  $\xx =\mf{ b,b^\bullet,c,d}$, we denote by $\mf{k}^\xx$ the contragredient Lie
(super)algebras (\cite[Section 2.5]{K1}) and by $\mf{k}^\mf{a}$  the Lie algebra $\mf{gl}(m)$. The corresponding Dynkin diagrams
\makebox(23,0){$\oval(20,12)$}\makebox(-20,8){$\mf{k}^\xx$} together with certain
distinguished sets of simple roots $\Pi(\mf{k^x})$ are listed as follows: \vspace{.3cm}

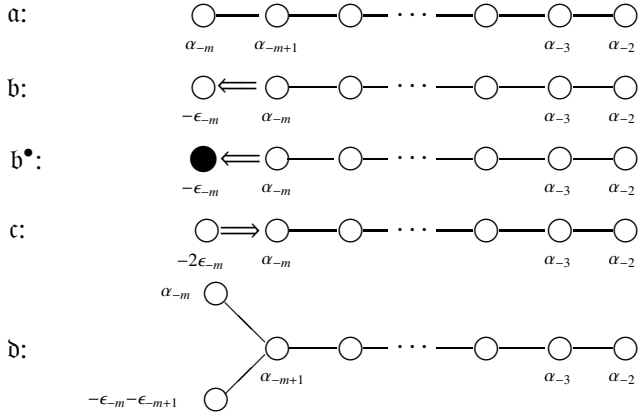
\begin{center}
\hskip -3cm \setlength{\unitlength}{0.16in}
\begin{picture}(24,3)
\put(8,2){\makebox(0,0)[c]{$\bigcirc$}}
\put(10.4,2){\makebox(0,0)[c]{$\bigcirc$}}
\put(14.85,2){\makebox(0,0)[c]{$\bigcirc$}}
\put(17.25,2){\makebox(0,0)[c]{$\bigcirc$}}
\put(19.4,2){\makebox(0,0)[c]{$\bigcirc$}}
\put(5.6,2){\makebox(0,0)[c]{$\bigcirc$}}
\put(8.4,2){\line(1,0){1.55}} \put(10.82,2){\line(1,0){0.8}}
\put(13.2,2){\line(1,0){1.2}} \put(15.28,2){\line(1,0){1.45}}
\put(17.7,2){\line(1,0){1.25}}
\put(6,2){\line(1,0){1.4}}
\put(12.5,1.95){\makebox(0,0)[c]{$\cdots$}}
\put(-.5,2){\makebox(0,0)[c]{$\mf{a}$:}}
\put(5.5,1){\makebox(0,0)[c]{\tiny$\alpha_{-m}$}}
\put(8,1){\makebox(0,0)[c]{\tiny$\alpha_{-m+1}$}}
\put(17.2,1){\makebox(0,0)[c]{\tiny$\alpha_{-3}$}}
\put(19.3,1){\makebox(0,0)[c]{\tiny$\alpha_{-2}$}}
\end{picture}
\end{center}

\begin{center}
\hskip -3cm \setlength{\unitlength}{0.16in}
\begin{picture}(24,2)
\put(8,2){\makebox(0,0)[c]{$\bigcirc$}}
\put(10.4,2){\makebox(0,0)[c]{$\bigcirc$}}
\put(14.85,2){\makebox(0,0)[c]{$\bigcirc$}}
\put(17.25,2){\makebox(0,0)[c]{$\bigcirc$}}
\put(19.4,2){\makebox(0,0)[c]{$\bigcirc$}}
\put(5.6,2){\makebox(0,0)[c]{$\bigcirc$}}
\put(8.4,2){\line(1,0){1.55}} \put(10.82,2){\line(1,0){0.8}}
\put(13.2,2){\line(1,0){1.2}} \put(15.28,2){\line(1,0){1.45}}
\put(17.7,2){\line(1,0){1.25}}
\put(6,1.8){$\Longleftarrow$}
\put(12.5,1.95){\makebox(0,0)[c]{$\cdots$}}
\put(-.5,2){\makebox(0,0)[c]{$\mf{b}$:}}
\put(5.5,1){\makebox(0,0)[c]{\tiny$-\epsilon_{-m}$}}
\put(8,1){\makebox(0,0)[c]{\tiny$\alpha_{-m}$}}
\put(17.2,1){\makebox(0,0)[c]{\tiny$\alpha_{-3}$}}
\put(19.3,1){\makebox(0,0)[c]{\tiny$\alpha_{-2}$}}
\end{picture}
\end{center}
\begin{center}
\hskip -3cm \setlength{\unitlength}{0.16in}
\begin{picture}(24,2)
\put(5.6,2){\circle*{0.9}} \put(8,2){\makebox(0,0)[c]{$\bigcirc$}}
\put(10.4,2){\makebox(0,0)[c]{$\bigcirc$}}
\put(14.85,2){\makebox(0,0)[c]{$\bigcirc$}}
\put(17.25,2){\makebox(0,0)[c]{$\bigcirc$}}
\put(19.4,2){\makebox(0,0)[c]{$\bigcirc$}}
\put(8.35,2){\line(1,0){1.5}} \put(10.82,2){\line(1,0){0.8}}
\put(13.2,2){\line(1,0){1.2}} \put(15.28,2){\line(1,0){1.45}}
\put(17.7,2){\line(1,0){1.25}}
\put(6.8,2){\makebox(0,0)[c]{$\Longleftarrow$}}
\put(12.5,1.95){\makebox(0,0)[c]{$\cdots$}}
\put(-.2,2){\makebox(0,0)[c]{$\mf{b^\bullet}$:}}
\put(5.5,1){\makebox(0,0)[c]{\tiny$-\epsilon_{-m}$}}
\put(8,1){\makebox(0,0)[c]{\tiny$\alpha_{-m}$}}
\put(17.2,1){\makebox(0,0)[c]{\tiny$\alpha_{-3}$}}
\put(19.3,1){\makebox(0,0)[c]{\tiny$\alpha_{-2}$}}
\end{picture}
\end{center}
\begin{center}
\hskip -3cm \setlength{\unitlength}{0.16in}
\begin{picture}(24,2)
\put(5.7,2){\makebox(0,0)[c]{$\bigcirc$}}
\put(8,2){\makebox(0,0)[c]{$\bigcirc$}}
\put(10.4,2){\makebox(0,0)[c]{$\bigcirc$}}
\put(14.85,2){\makebox(0,0)[c]{$\bigcirc$}}
\put(17.25,2){\makebox(0,0)[c]{$\bigcirc$}}
\put(19.4,2){\makebox(0,0)[c]{$\bigcirc$}}
\put(6.8,2){\makebox(0,0)[c]{$\Longrightarrow$}}
\put(8.4,2){\line(1,0){1.55}} \put(10.82,2){\line(1,0){0.8}}
\put(13.2,2){\line(1,0){1.2}} \put(15.28,2){\line(1,0){1.45}}
\put(17.7,2){\line(1,0){1.25}}
\put(12.5,1.95){\makebox(0,0)[c]{$\cdots$}}
\put(-.5,2){\makebox(0,0)[c]{$\mf{c}$:}}
\put(5.5,1){\makebox(0,0)[c]{\tiny$-2\epsilon_{-m}$}}
\put(8,1){\makebox(0,0)[c]{\tiny$\alpha_{-m}$}}
\put(17.2,1){\makebox(0,0)[c]{\tiny$\alpha_{-3}$}}
\put(19.3,1){\makebox(0,0)[c]{\tiny$\alpha_{-2}$}}
\end{picture}
\end{center}
\begin{center}
\hskip -3cm \setlength{\unitlength}{0.16in}
\begin{picture}(24,3.5)
\put(8,2){\makebox(0,0)[c]{$\bigcirc$}}
\put(10.4,2){\makebox(0,0)[c]{$\bigcirc$}}
\put(14.85,2){\makebox(0,0)[c]{$\bigcirc$}}
\put(17.25,2){\makebox(0,0)[c]{$\bigcirc$}}
\put(19.4,2){\makebox(0,0)[c]{$\bigcirc$}}
\put(6,3.8){\makebox(0,0)[c]{$\bigcirc$}}
\put(6,.3){\makebox(0,0)[c]{$\bigcirc$}}
\put(8.4,2){\line(1,0){1.55}} \put(10.82,2){\line(1,0){0.8}}
\put(13.2,2){\line(1,0){1.2}} \put(15.28,2){\line(1,0){1.45}}
\put(17.7,2){\line(1,0){1.25}}
\put(7.6,2.2){\line(-1,1){1.3}}
\put(7.6,1.8){\line(-1,-1){1.3}}
\put(12.5,1.95){\makebox(0,0)[c]{$\cdots$}}
\put(-.5,2){\makebox(0,0)[c]{$\mf{d}$:}}
\put(3.3,0.3){\makebox(0,0)[c]{\tiny${-}\epsilon_{-m}{-}\epsilon_{-m+1}$}}
\put(4.7,3.8){\makebox(0,0)[c]{\tiny$\alpha_{-m}$}}
\put(8.2,1){\makebox(0,0)[c]{\tiny$\alpha_{-m+1}$}}
\put(17.2,1){\makebox(0,0)[c]{\tiny$\alpha_{-3}$}}
\put(19.3,1){\makebox(0,0)[c]{\tiny$\alpha_{-2}$}}
\end{picture}
\end{center}
According to \cite[Proposition 2.5.6]{K1} these Lie (super)algebras are $\mf{so}(2m+1)$,
$\mf{osp}(1|2m)$, $\mf{sp}(2m)$ for $m\ge 1$ and $\mf{so}(2m)$ for $m \ge 2$, for $\xx =\mf{ b,b^\bullet,c,d}$,
respectively.  We will use the same notation
\makebox(23,0){$\oval(20,12)$}\makebox(-20,8){$\mf{k}^\xx$} to denote the diagrams of all
the degenerate cases for $m=0,1$ as well.  We have used \makebox(15,5){\circle*{7}} to denote an odd
non-isotropic simple root.

For $n\in\N$, let
\makebox(23,0){$\oval(20,12)$}\makebox(-20,8){$\mf{T}_n$},
\makebox(23,0){$\oval(20,14)$}\makebox(-20,8){$\ov{\mf{T}}_n$} and
\makebox(23,0){$\oval(20,14)$}\makebox(-20,8){$\wt{\mf{T}}_n$}
denote the following Dynkin diagrams, where $\bigotimes$ denotes
an odd isotropic simple root:
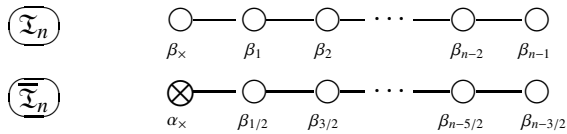
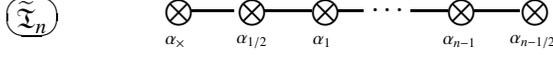
\begin{center}
\hskip -3cm \setlength{\unitlength}{0.16in}
\begin{picture}(24,3)
\put(8,2){\makebox(0,0)[c]{$\bigcirc$}}
\put(10.4,2){\makebox(0,0)[c]{$\bigcirc$}}
\put(14.85,2){\makebox(0,0)[c]{$\bigcirc$}}
\put(17.25,2){\makebox(0,0)[c]{$\bigcirc$}}
\put(5.6,2){\makebox(0,0)[c]{$\bigcirc$}}
\put(8.4,2){\line(1,0){1.55}} \put(10.82,2){\line(1,0){0.8}}
\put(13.2,2){\line(1,0){1.2}} \put(15.28,2){\line(1,0){1.45}}
\put(6,2){\line(1,0){1.4}}
\put(12.5,1.95){\makebox(0,0)[c]{$\cdots$}}
\put(0,1.2){{\ovalBox(1.6,1.2){$\mf{T}_n$}}}
\put(5.5,1){\makebox(0,0)[c]{\tiny$\beta_{\times}$}}
\put(8,1){\makebox(0,0)[c]{\tiny$\beta_{1}$}}
\put(10.3,1){\makebox(0,0)[c]{\tiny$\beta_{2}$}}
\put(15,1){\makebox(0,0)[c]{\tiny$\beta_{n-2}$}}
\put(17.2,1){\makebox(0,0)[c]{\tiny$\beta_{n-1}$}}
\end{picture}
\end{center}
\begin{center}
\hskip -3cm \setlength{\unitlength}{0.16in}
\begin{picture}(24,2)
\put(8,2){\makebox(0,0)[c]{$\bigcirc$}}
\put(10.4,2){\makebox(0,0)[c]{$\bigcirc$}}
\put(14.85,2){\makebox(0,0)[c]{$\bigcirc$}}
\put(17.25,2){\makebox(0,0)[c]{$\bigcirc$}}
\put(5.6,2){\makebox(0,0)[c]{$\bigotimes$}}
\put(8.4,2){\line(1,0){1.55}} \put(10.82,2){\line(1,0){0.8}}
\put(13.2,2){\line(1,0){1.2}} \put(15.28,2){\line(1,0){1.45}}
\put(6,2){\line(1,0){1.4}}
\put(12.5,1.95){\makebox(0,0)[c]{$\cdots$}}
\put(0,1.2){{\ovalBox(1.6,1.2){$\ov{\mf{T}}_n$}}}
\put(5.5,1){\makebox(0,0)[c]{\tiny$\alpha_{\times}$}}
\put(10.3,1){\makebox(0,0)[c]{\tiny$\beta_{3/2}$}}
\put(8,1){\makebox(0,0)[c]{\tiny$\beta_{1/2}$}}
\put(14.8,1){\makebox(0,0)[c]{\tiny$\beta_{n-5/2}$}}
\put(17.5,1){\makebox(0,0)[c]{\tiny$\beta_{n-3/2}$}}
\end{picture}
\end{center}
\begin{center}
\hskip -3cm \setlength{\unitlength}{0.16in}
\begin{picture}(24,2)
\put(8,2){\makebox(0,0)[c]{$\bigotimes$}}
\put(10.4,2){\makebox(0,0)[c]{$\bigotimes$}}
\put(14.85,2){\makebox(0,0)[c]{$\bigotimes$}}
\put(17.25,2){\makebox(0,0)[c]{$\bigotimes$}}
\put(5.6,2){\makebox(0,0)[c]{$\bigotimes$}}
\put(8.4,2){\line(1,0){1.55}} \put(10.82,2){\line(1,0){0.8}}
\put(13.2,2){\line(1,0){1.2}} \put(15.28,2){\line(1,0){1.45}}
\put(6,2){\line(1,0){1.4}}
\put(12.5,1.95){\makebox(0,0)[c]{$\cdots$}}
\put(0,1.2){{\ovalBox(1.6,1.2){$\wt{\mf{T}}_n$}}}
\put(5.5,1){\makebox(0,0)[c]{\tiny$\alpha_{\times}$}}
\put(8,1){\makebox(0,0)[c]{\tiny$\alpha_{1/2}$}}
\put(14.8,1){\makebox(0,0)[c]{\tiny$\alpha_{n-1}$}}
\put(17.2,1){\makebox(0,0)[c]{\tiny$\alpha_{n-1/2}$}}
\put(10.3,1){\makebox(0,0)[c]{\tiny$\alpha_{1}$}}
\end{picture}
\end{center}
We will denote the sets of simple roots of the above diagrams accordingly by
$\Pi({\mf{T}}_n)$, $\Pi(\ov{\mf{T}}_n)$ and $\Pi(\wt{\mf{T}}_n)$, respectively.
The Lie superalgebras associated with these Dynkin diagrams are
$\gl(n+1)$, $\mf{gl}(1|n)$ and $\mf{gl}(n|n+1)$, respectively.
In the
limit $n\to\infty$, the associated Lie superalgebras are direct limits of
these Lie superalgebras, and we
will simply drop $\infty$ to write
\makebox(23,0){$\oval(20,12)$}\makebox(-20,8){$\mf{T}$} =
\makebox(23,0){$\oval(20,12)$}\makebox(-20,8){$\mf{T}_\infty$} and
so on.

Any of the {\em head} diagrams
\makebox(23,0){$\oval(20,11)$}\makebox(-20,8){$\mf{k}^\xx$} may be
connected with the {\em tail} diagrams
\makebox(23,0){$\oval(20,14)$}\makebox(-20,8){$\wt{\mf{T}}_n$},
\makebox(23,0){$\oval(20,14)$}\makebox(-20,8){${\mf{T}}_n$} and
\makebox(23,0){$\oval(20,14)$}\makebox(-20,8){$\ov{\mf{T}}_n$} to
produce the following Dynkin diagrams ($n\in\N\cup\{\infty\}$):
\begin{equation}\label{Dynkin:combined}
\hskip -3cm \setlength{\unitlength}{0.16in}
\begin{picture}(24,1)
\put(5.0,0.5){\makebox(0,0)[c]{{\ovalBox(1.6,1.2){$\mf{k}^\xx$}}}}
\put(5.8,0.5){\line(1,0){1.85}}
\put(8.5,0.5){\makebox(0,0)[c]{{\ovalBox(1.6,1.2){$\wt{\mf{T}}_n$}}}}
\put(15,0.5){\makebox(0,0)[c]{{\ovalBox(1.6,1.2){$\mf{k}^\xx$}}}}
\put(15.8,0.5){\line(1,0){1.85}}
\put(18.5,0.5){\makebox(0,0)[c]{{\ovalBox(1.6,1.2){${\mf{T}}_n$}}}}
\put(25,0.5){\makebox(0,0)[c]{{\ovalBox(1.6,1.2){$\mf{k}^\xx$}}}}
\put(25.8,0.5){\line(1,0){1.85}}
\put(28.5,0.5){\makebox(0,0)[c]{{\ovalBox(1.6,1.2){$\ov{\mf{T}}_n$}}}}
\end{picture}
\end{equation}
We will denote the sets of simple roots of the above diagrams accordingly by $\wt\Pi^{\xx}_n$, $\Pi^{\xx}_n$ and $\ov\Pi^{\xx}_n$, respectively.
Let us denote the three Dynkin diagrams of \eqnref{Dynkin:combined}
by
\makebox(23,0){$\oval(20,14)$}\makebox(-20,8){\small{$\w{\mc{G}}^\xx_n$}},
\makebox(23,0){$\oval(20,14)$}\makebox(-20,8){\small${\mc{G}}^\xx_n$} and
\makebox(23,0){$\oval(20,14)$}\makebox(-20,8){\small$\ov{\mc{G}}^\xx_n$} for $n\in\mathbb{N}\cup\{\infty\}$.
We will drop the subscript $\infty$ for $n=\infty$.
For Dynkin diagrams of the degenerate cases, we refer to \cite[Section 2.3]{CLW} for details.

\subsection{Central extensions}\label{Central Ext}
 Consider the central extension
$\DG^\xx_n$ (resp., $\G^\xx_n$ and $\SG^\xx_n$) of $\mGt^\xx_n$ (resp., $\mG^\xx_n$ and $\mGb^\xx_n$) for $\xx=\mf{a, b, b^\bullet, c, d}$
by the one-dimensional center $\C K$ determined by the $2$-cocycle
\begin{align}\label{central extension}
\tau(A,B):=\text{Str}([J,A]B),\quad \hbox{$A,B\in \mGt^\xx_n$ (resp., $\mG^\xx_n$ and $\mGb^\xx_n$)},
\end{align}
where $J:=-\sum_{r\ge \hf} E_{rr}$ and $\text{Str}$ denotes the supertrace defined by
 $$
 \mbox{Str}\big((a_{i,j})_{i,j\in\w{\I}_m}\big):=\sum_{j\in\w{\I}_m}(-1)^{|v_j|}a_{j,j}.
 $$
In fact, the cocycle
$\tau$ is a coboundary. Moreover, there is an isomorphism $\iota$ from the direct sum of Lie superalgebras $\mGt^\xx_n\oplus \C K$ (resp., $\mG^\xx_n\oplus \C K$ and $\mGb^\xx_n\oplus \C K$) to $\DG^\xx_n$ (resp., $\G^\xx_n$ and $\SG^\xx_n$) defined by
\begin{equation}\label{iso}
  \iota(A)=A+\mbox{Str}(JA)K,\quad\hbox{for $A\in\mGt^\xx_n$ (resp., $\mG^\xx_n$ and $\mGb^\xx_n$)\quad and\quad $\iota(K)=K$.}
\end{equation}

\begin{rem} Using the isomorphisms given in \eqref{iso} and the standard bases of the Cartan subalgebras, it is easy to see that the central extensions defined by the $2$-cycle in \eqnref{central extension} are the same as the central extensions defined in \cite[Section 2.4]{CLW} for $\xx=\mf{b, b^\bullet, c, d}$.
\end{rem}

Every $\DG^\xx_n$ (resp., $\G^\xx_n$ and $\SG^\xx_n$)-module can be regarded as a $\mGt^\xx_n$ (resp., $\mG^\xx_n$ and $\mGb^\xx_n)$-module through the isomorphism \eqnref{iso} since $\mGt^\xx_n$ (resp., $\mG^\xx_n$ and $\mGb^\xx_n)$ is a subalgebra of $\mGt^\xx_n\oplus \C K$ (resp., $\mG^\xx_n\oplus \C K$ and $\mGb^\xx_n\oplus \C K$).  These central extensions are convenient and conceptual for the formulation of truncation functors and super duality described this section below (see \cite[Remark 3.3]{CLW} for more explanations).

Note that $\w{\mc{G}}^\xx_n$, ${\mc{G}}^\xx_n$ and $\ov{\mc{G}}^\xx_n$ are naturally
subalgebras of $\mf{gl}(\w{V})$ and hence $\DG^\xx_n$, $\G^\xx_n$ and
$\SG^\xx_n$ are naturally subalgebras of the central extension of $\mf{gl}(\w{V})$.
Also note that ${\mc{G}}^\xx_n$ and $\ov{\mc{G}}^\xx_n$ are subalgebras of $\w{\mc{G}}^\xx_n$ and $\G^\xx_n$ and
$\SG^\xx_n$ are subalgebras of $\DG^\xx_n$.  The standard Cartan subalgebras of $\G^\xx_n$, $\SG^\xx_n$ and
$\DG^\xx_n$ will be denoted by $\h^\xx_n$, $\ov{\h}^\xx_n$ and $\wt{\h}^\xx_n$, respectively.
They have bases $\{K,{E}^{\xx}_{r} \}$  with dual
bases $\{\Lambda_0,\epsilon_r\}$ in the restricted dual $(\wt{\h}^\xx_n)^*$, $({\h}^\xx_n)^*$
and $(\wt{\h}^\xx_n)^*$,  where $r$ runs over the index sets $\wt{\I}^+_m(n)$, ${\I}^+_m(n)$ and
$\ov{\I}^+_m(n)$, respectively. Here
$E^{\xx}_r\in \wt{\h}^\xx_n$ (resp., ${\h}^\xx_n$ and $\ov{\h}^\xx_n$)
and $\La_0\in (\wt{\h}^\xx_n)^*$  (resp., $({\h}^\xx_n)^*$ and $(\ov{\h}^\xx_n)^*$) are defined by
\begin{equation*}\label{def:E}
E^{\xx}_r:=\begin{cases}
     E_{rr} & \mbox{for}\ \xx=\mf{a};\\
    E_{rr}-E_{\ov{r}\ov{r}} & \mbox{for}\ \xx=\mf{b,b^\bullet,c,d},
     \end{cases}
     \quad \hbox{and}\quad \La_0(K)=1,\quad\La_0({E^{\xx}_r})=0,
\end{equation*}
for all $r \in \wt{\I}^+_m(n)$
 (resp., ${\I}^+_m(n)$ and $\ov{\I}^+_m(n)$).

{\bf In the remainder of the paper we shall drop the superscript
$\xx$ and the symbol $\infty$ if it causes no ambiguity.} For example, we write $\mGt_n$, $\mG_n$ and $\mGb_n$ for
$\mGt^\xx_n$, $\mG^\xx_n$ and $\mGb^\xx_n$, and $\wt{\G}_n$, $\G_n$ and $\SG_n$ for
$\wt{\G}^\xx_n$, $\G^\xx_n$ and $\SG^\xx_n$, respectively, where $\xx$ denotes a fixed type among
$\mf{a, b,b^\bullet,c,d}$. Also, we write $\wt{\G}$, $\G$ and $\SG$ for
$\wt{\G}^\xx_\infty$, $\G^\xx_\infty$ and $\SG^\xx_\infty$, respectively.

\subsection{Categories $\w{\mc {O}}_n$, ${\mc {O}}_n$ and $\ov{\mc {O}}_n$}

Recall that $\wt\Pi^{\xx}_n$, $\Pi^{\xx}_n$ and $\ov\Pi^{\xx}_n$ denote the sets of simple roots of $\DG_n$, $\G_n$ and $\SG_n$, respectively. Let $\w{\Phi}_n^{+}$ (resp., $\Phi_n^{+}$ and $\ov{\Phi}_n^{+}$) denote the set of positive roots and let
 $\w{\mathfrak{b}}_n$ (resp., $\mathfrak{b}_n$ and $\overline{\mathfrak{b}}_n$) denote the corresponding Borel subalgebra of
$\w{\mathfrak{g}}_n$ (resp., $ \mathfrak{g}_n$ and $\overline{\mathfrak{g}}_n$).

For $n\in \N\cup\{\infty\}$, let $\w Y_n$, $ Y_n$ and $\ov Y_n$ denote the subsets of $\w\Pi_n$, $\Pi_n$ and $\ov\Pi_n$, respectively, defined by
\begin{equation}\label{Y}
 \w Y_n:= \Pi(\mf{\w T}_n)\setminus\{\alpha_{\times}\},\quad
  Y_n:=\Pi(\mf{T}_n)\setminus\{\beta_{\times}\}\quad \hbox{and}\quad \overline Y_n:=\Pi(\mf{\overline T}_n)\setminus\{\alpha_{\times}\} .
\end{equation}
Let $\w{\mf l}_n$ (resp., $\mf{l}_n$ and $\ov{\mf{l}}_n$) be the standard Levi subalgebra of $ \DG_n$ (resp., $\G_n$ and $\SG_n$) associated to $\w Y_n$ (resp., $ Y_n$ and $\ov Y_n$) and let  $\w{\mf p}_n=\w{\mf l}_n+\w{\mf b}_n$
(resp., $\mf p_n=\mf l_n+\mf b_n$ and $\overline{\mf p}_n=\overline{\mf l}_n+\overline{\mf b}_n$) be the corresponding parabolic subalgebra.

Given a partition $\mu=(\mu_1,\mu_2,\ldots)$, we denote by
$\ell(\mu)$ the length of $\mu$ and by $\mu'$ its conjugate
partition. We also denote by $\theta(\mu)$ the modified Frobenius
coordinates of $\mu$:
\begin{equation*}
\theta(\mu)
:=(\theta(\mu)_{1/2},\theta(\mu)_1,\theta(\mu)_{3/2},\theta(\mu)_2,\ldots),
\end{equation*}
where
$$\theta(\mu)_{i-1/2}:=\max\{\mu'_i-i+1,0\}, \quad
\theta(\mu)_i:=\max\{\mu_i-i, 0\}, \quad i\in\N.
$$

Associated to a partition $\la^+=(\la^+_1,\la^+_2,\ldots)$, $d\in\C$ and $\la_{-m},\ldots,\la_{-1}\in\C$ (resp., $\la_{-m},\ldots,\la_{-1}\in\Z$) for $\xx=\mf{a, b, c, d}$ (resp., $\mf b^\bullet$), we define
\begin{align}
{\la} &:=\sum_{i=-m}^{-1}\la_{i}\epsilon_{i}
 + \sum_{j\in\N}\la^+_{j}\epsilon_{j}
 + d\La_0\in ({\h}^\xx)^{*},\label{weight:Im}\\
\ov{\la} &:=\sum_{i=-m}^{-1}\la_{i}\epsilon_{i}
 + \sum_{s\in\hf+\Z_+}(\la^+)'_{s+\hf}\epsilon_s
 + d\La_0\in (\ov{\h}^\xx)^{*}, \label{weight:ovIm}\\
\w{\la} &:=\sum_{i=-m}^{-1}\la_{i}\epsilon_{i}
 + \sum_{r\in\hf\N}\theta(\la^+)_r\epsilon_r
 + d\La_0\in (\wt{\h}^\xx)^{*},\label{weight:wtIm}
\end{align}
which are called \emph{dominant weights}.
Let $P^+:=P^+_\infty\subset{\h}^*$, $\ov{P}^+:=\ov{P}^+_\infty\subset\ov{\h}^*$ and
$\wt{P}^+:=\wt{P}^+_\infty\subset\wt{\h}^*$ denote the sets of all dominant weights of the forms
\eqnref{weight:Im}, \eqnref{weight:ovIm} and \eqnref{weight:wtIm},
respectively. The purpose of the definitions of the dominant weights for  $\mf b^\bullet$ is to give a simple way to define the $\Z_2$-gradations for the $\w{\mf g}^{\mf b^\bullet}$-modules, $\mf g^{\mf b^\bullet}$-modules and $\overline{\mf g}^{\mf b^\bullet}$-modules in the categories defined below. By definition we have bijective maps
\begin{equation*}
\begin{array}{l}
-:P^+\longrightarrow\bar{P}^+\\
  \ \ \ \   \ \ \  \la\ \mapsto
\ \ov\la
\end{array} \ \ \ \ \  \ \ \mbox{and}\ \ \ \ \ \
\begin{array}{l}
 \sim:P^+\longrightarrow\wt{P}^+\\
 \ \ \ \   \ \   \la\ \mapsto\ \w\la
\end{array}.
\end{equation*}

Given $n\in\N$ and $\la\in P^+$ with $\la^+_j=0$ (resp., $(\la^+)'_j=0$ and $\theta(\la^+)_j=0$) for
$j>n$, we may regard it as a weight $\la_n \in \h^*_n$ (resp., $\ov{\la}_n \in
\ov{\h}^*_n$ and $\w{\la}_n \in
\wt{\h}^*_n$) in a
natural way. The subsets of such weights $\la_n, \ov{\la}_n,
\w{\la}_n$ in $\h^*_n$, $\ov{\h}^*_n$ and $\wt{\h}^*_n$ will be
denoted by $P^+_n$, $\bar{P}^+_n$ and $\wt{P}^+_n$, respectively.

 For $n\in\N\cup\{\infty\}$ and $\mu\in \wt{P}^+_n$  (resp., $P^+_n$ and $\ov{P}^+_n$), we denote by $\w{\Delta}_n(\mu)=\mbox{Ind}^{\w{\mf g}_n}_{\w{\mf p}_n} L(\w{\mf l}_n, \mu)$
 (resp., $\Delta_n(\mu)=\mbox{Ind}^{\mf g_n}_{\mf p_n}L(\mf l_n, \mu)$ and
 $\ov{\Delta}_n(\mu)=\mbox{Ind}^{\ov{\mf g}_n}_{\ov{\mf p}_n} L(\ov{\mf l}_n, \mu)$)
 the parabolic Verma $\w{\mf g}_n$- (resp., $\mf g_n$- and $\overline{\mf g}_n$-)module,
 where $ L(\w{\mf l}_n, \mu)$ (resp., $L(\mf l_n, \mu)$ and $ L(\overline{\mf l}_n, \mu)$)
 is the irreducible highest weight $\w{\mf l}_n$- (resp., $\mf l_n$- and $\overline{\mf l}_n$-)module of highest weight
 $\mu$.
 The unique irreducible quotient $\w{\mf g}_n$- (resp., $\mf g_n$- and $\overline{\mf g}_n$-)module of
 $\w{\Delta}_n(\mu)$ (resp., $\Delta_n(\mu)$ and $\ov{\Delta}_n(\mu)$)
 is denoted by
 $\w L_n(\mu)$ (resp., $L_n(\mu)$ and $\overline L_n(\mu)$).

 For $n\in\N\cup\{\infty\}$ , let $\w{\mathcal O}_n$ (resp., $\mathcal O_n$ and $\ov{\mathcal O}_n$) be the category of $\w{\mf g}_n$- (resp., $\mf g_n$- and $\overline{\mf g}_n$-)modules $M$ such that $M$ is a semisimple $\w{\mf h}_n$- (resp., $\mf h_n$- and $\overline{\mf h}_n$-)module with finite dimensional weight subspaces $M_{\gamma}$ for $\gamma\in\w{\mf h}_n^*$ (resp., $\mf h_n^*$ and $\overline{\mf h}_n^*$), satisfying
   \begin{itemize}
     \item[ (i)] $M$ decomposes over $\w{\mf l}_n$ (resp., $\mf l_n$ and $\overline{\mf l}_n$)  as a direct sum of $\w L(\w{\mf l}_n, \mu)$ (resp., $L(\mf l_n, \mu)$ and $\overline L(\overline{\mf l}_n, \mu)$) for $\mu\in \wt{P}^+_n$  (resp., $P^+_n$ and $\bar{P}^+_n$).
     \item [(ii)] There exist finitely many weights $\lambda^1,\lambda^2,\ldots,\lambda^k\in \wt{P}^+_n$  (resp., $P^+_n$ and $\bar{P}^+_n$) (depending on $M$) such that if $\gamma$ is a weight in $M$, then $\lambda^i-\gamma\in \sum_{\alpha\in\w\Pi_n}\mathbb{Z}_+\alpha$ (resp., $\sum_{\alpha\in\Pi_n}\mathbb{Z}_+\alpha$ and $\sum_{\alpha\in\ov\Pi_n}\mathbb{Z}_+\alpha$) for some $i$.
   \end{itemize}
The morphisms in the categories are even homomorphisms of modules. Our categories $\w{\mc {O}}_n$, ${\mc {O}}_n$ and $\ov{\mc {O}}_n$ are the largest categories among the categories given in \cite[Section 6.2]{CW2}. The categories are abelian categories.
There is a natural $\Z_2$-gradation on each module in the categories with compatible action of the corresponding Lie (super)algebra defined as follows.  Set
 \begin{align}\label{weight}
 \wt{\Xi}_n^\xx&:=\begin{cases}
 \C\Lambda_0+\sum_{j=-m}^{-1}\C\epsilon_i+\sum_{r\in\w{\I}^+_0(n)}\Z_+\epsilon_r,& \mbox{if}\,\,\, \xx=\mf{a, b, c, d}; \\
 \C\Lambda_0+\sum_{j=-m}^{-1}\Z\epsilon_i+\sum_{r\in\w{\I}^+_0(n)}\Z_+\epsilon_r,& \mbox{if}\,\,\, \xx=\mf{b}^\bullet,
 \end{cases}\nonumber\\
 {\Xi}_n^\xx&:=\begin{cases}
 \C\Lambda_0+\sum_{j=-m}^{-1}\C\epsilon_i+\sum_{r\in{\I}^+_0(n)}\Z_+\epsilon_r,& \mbox{if}\,\,\, \xx=\mf{a, b, c, d};\\
 \C\Lambda_0+\sum_{j=-m}^{-1}\Z\epsilon_i+\sum_{r\in{\I}^+_0(n)}\Z_+\epsilon_r,& \mbox{if}\,\,\, \xx=\mf{b}^\bullet,
 \end{cases}\\
 \ov{\Xi}_n^\xx&:=\begin{cases}
 \C\Lambda_0+\sum_{j=-m}^{-1}\C\epsilon_i+\sum_{r\in\ov{\I}^+_0(n)}\Z_+\epsilon_r,& \mbox{if}\,\,\, \xx=\mf{a, b, c, d};\\
 \C\Lambda_0+\sum_{j=-m}^{-1}\Z\epsilon_i+\sum_{r\in\ov{\I}^+_0(n)}\Z_+\epsilon_r,& \mbox{if}\,\,\, \xx=\mf{b}^\bullet.
 \end{cases}\nonumber
 \end{align}
For $\varepsilon=0$ or $1$ and $\Theta=\wt{\Xi}$, ${\Xi}$ or $\ov{\Xi}$, we define
\begin{equation*}
{\Theta}_n^\xx(\varepsilon):=\begin{cases}
\{\mu\in {\Theta}_n^\xx\,|\, \sum_{r=1}^n\mu(E_{r+\hf})\equiv \varepsilon \,\,(\text{mod }2)\,\},& \mbox{if}\,\,\, \xx=\mf{a, b, c, d};\\
\{\mu\in {\Theta}_n^\xx\,|\, \sum_{r\in{\I}^+_m(n)}\mu(E_r)\equiv \varepsilon \,\,(\text{mod }2)\,\},& \mbox{if}\,\,\, \xx=\mf{b}^\bullet,
\end{cases}
\end{equation*}
As before, we will drop the subscript $n$ for $n=\infty$ and the superscript $\xx$. It is clear that the weights of $M$ are contained in ${\Xi}_n$ and $\ov{\Xi}_n$ for $M\in \mathcal O_n$ and $\ov{\mathcal O}_n$, respectively. By the paragraph before Theorem 6.4 in \cite{CW2}, the weights of $M$ are contained in $\w{\Xi}_n$ for  $M\in \w{\mathcal O}_n$. For $M\in \w{\mathcal O}_n$, ${M}={M}_{\ov{0}}\bigoplus {M}_{\ov{1}}$  is a $\Z_2$-graded vector space such that
 \begin{equation}\label{wt-Z2-gradation}
{M}_{\ov{0}}:=\bigoplus_{\mu\in\w{\Xi}_n(0)}{M}_{\mu}\qquad\hbox{and}\qquad
{M}_{\ov{1}}:=\bigoplus_{\mu\in\w{\Xi}_n(1)}{M}_{\mu}.
 \end{equation}
By the description of the set of positive roots for $\DG_n$ given in \cite[Section 6.1.3, 6.1.4]{CW2}, the $\Z_2$-gradation on $M$ is compatible with the action of $\DG_n$.
Similarly, we define a $\Z_2$-gradation with compatible  action of ${\mf g}_n$ and $\ov{\mf g}_n$ on ${M}$ for $M\in\mathcal O_n$ and $\ov{\mathcal O}_n$, respectively. The categories are denoted by $\w{\mc {O}}^{\ov 0}$, ${\mc {O}}^{\ov 0}$ and $\ov{\mc {O}}^{\ov 0}$ in \cite{CL} and \cite{CLW}. It is clear that $\mathcal O_n$ and $\ov{\mathcal O}_n$ are tensor categories. By \cite[Theorem 6.4]{CW2}, $\w{\mathcal O}_n$ is also tensor category. Note that the $\Z_2$-gradation on $M\otimes N$ given by \eqnref{wt-Z2-gradation} and the $\Z_2$-gradation on $M\otimes N$ induced from the $\Z_2$-gradations on $M$ and $N$ given by
\eqref{wt-Z2-gradation} are the same for $M, N\in \w{\mathcal O}_n$ (resp., $\mathcal O_n$ and $\ov{\mathcal O}_n$).

A proof of the following proposition is given in \cite[Corollary 6.8]{CW2}.
\begin{prop}
Let $n\in\N\cup\{\infty\}$.
\begin{itemize}
  \item[(i)] The modules $\w{\Delta}_n(\la)$ and $\w{L}_n(\la)$ lie in  $\w{\mathcal O}_n$ for all $\la \in \w P^+_n$,
  \item[(ii)]  The modules ${\Delta}_n(\la)$ and ${L}_n(\la)$ lie in  $\w{\mathcal O}_n$ for all $\la \in P^+_n$,
  \item[(iii)]   The modules $\ov{\Delta}_n(\la)$ and $\ov{L}_n(\la)$ lie in  $\ov{\mathcal O}_n$ for all $\la \in \ov P^+_n$.
\end{itemize}
\end{prop}

The following lemma is easy to see by using the weights of the modules described in \eqnref{weight}.

\begin{lem}\label{weight=0}
 For $n\in\N\cup\{\infty\}$, let $M, N\in \w{\mathcal O}_n$ (resp., $\mathcal O_n$ and $\ov{\mathcal O}_n$) and let $\mu$ and $\gamma$ be weights of $M$ and $N$, respectively. We have
\[
(\mu+\gamma)(E_r)=0\quad\hbox{if and only if }\quad \mu(E_r)=0\,\, \hbox{and}\,\, \gamma(E_r)=0,
 \]
for $r\in \wt{\I}^+_0(n)$ (resp., ${\I}^+_0(n)$ and
$\ov{\I}^+_0(n)$).
\end{lem}

 For $0\le k< n\leq\infty$, the truncation functor
$\mf{tr}^n_k: \mathcal {O}_n\longrightarrow\mathcal {O}_k$ is defined by
\[
  \mf{tr}^n_k(M)=\bigoplus_{\nu\in \Xi_k}M_{\nu} \qquad \hbox{for $M\in \OO_n$}
\]
and $\mf{tr}^n_k(f)$  is defined to be the restriction of $f$ to  $\mf{tr}^n_k(M)$ for $f \in {\rm Hom}_{\OO_n}(M, N)$. Similarly,
$\ov{\mf{tr}}^n_k: \overline{\mathcal {O}}_n\longrightarrow\overline{\mathcal {O}}_k$ and
 $\w{\mf{tr}}^n_k: \w{\mathcal {O}}_n\longrightarrow\w{\mathcal {O}}_k$  are defined.  It is clear that $\w{\mf{tr}}^n_k$, ${\mf{tr}}^n_k$ and $\ov{\mf{tr}}^n_k$ are exact functors. By \lemref{weight=0}, we have $\w{\mf{tr}}^n_k(M\otimes N)=\w{\mf{tr}}^n_k(M)\otimes\w{\mf{tr}}^n_k(N)$ for all $M, N\in \w{\mathcal O}_n$ and hence $\w{\mf{tr}}^n_k$ is a tensor functor. Similarly, ${\mf{tr}}^n_k$ and $\ov{\mf{tr}}^n_k$ are tensor functors.
A proof of the following proposition is given in \cite[Proposition 6.9]{CW2}.

 \begin{prop}\label{thm:truncation}
 Let $0\le k<n \leq\infty$ and $V=\w{L},\w{\Delta}$. For ${\mu}\in \w{P}^+_n$, we have
 $$
 \w{\mf{tr}}^n_k( {V}_n({\mu}))=\begin{cases}
       {V}_k({\mu}),& \mbox{if}\ \ {\mu}\in \w{P}^+_k;\\
       0,& \mbox{otherwise}.
     \end{cases}
 $$
 Similar statements hold for ${\mf{tr}}^n_k$ and
 $\ov{\mf{tr}}^n_k$.
   \end{prop}

The following lemma follows from Lemma \ref{weight=0}.

\begin{lem}\label{weight=Xi}
Let $\mu, \gamma\in \w\Xi$. We have
\[
\mu+\gamma \in \Xi,\, \hbox{(resp., $\ov\Xi$)} \quad\hbox{if and only if }\quad \mu,\,\gamma \in \Xi\,\hbox{(resp., $\ov\Xi$)}.
 \]

\end{lem}

Now we recall the functors $T:\w{\mc{O}}\rightarrow{\mc{O}}$ and $\ov{T}:\w{\mc{O}}\rightarrow\ov{\mc{O}}$ defined in \cite{CL} and \cite{CLW} as follows.
Given ${M}=\bigoplus_{\gamma\in\w{\h}^*}{M}_\gamma\in\w{\mc{O}}$,  we define
\begin{align}\label{def:T}
T({M}):= \bigoplus_{\gamma\in{{\h}^*}}{M}_\gamma,\qquad
\hbox{and}\qquad \ov{T}({M}):=
\bigoplus_{\gamma\in{\ov{\h}^*}}{M}_\gamma.
\end{align}
It is clear that $T({M})$ is a ${\G}$-module and $\ov{T}({M})$ is a $\ov{\G}$-module. For $M, N\in\w{\mc{O}}$ and $f\in {\rm Hom}_{\w\OO}(M,N)$, we define $T(f)$ and $\ov T({f})$ to be the restrictions of $f$ to $T({M})$ and $\ov{T}({M})$, respectively.
It is also clear  that $f( T({M}))\subseteq T({N})$ and  $f( \ov T({M}))\subseteq \ov T( {N})$, and
\[
T({ f}):\ T({M})\ \longrightarrow\ T({N})\qquad \hbox{and }\qquad \ov T({f}):\ \ov T({M})\ \longrightarrow\ \ov T({N})
\]
are $\G$-homomorphism and $\SG$-homomorphism, respectively.
The functors $T$ and $\ov T$ are exact (see, for example, \cite[Proposition 6.15]{CW2}).
By \lemref{weight=Xi}, we have ${T}(M\otimes N)={T}(M)\otimes{T}(N)$ and  $\ov{T}(M\otimes N)=\ov{T}(M)\otimes\ov{T}(N)$ for all $M, N\in \w{\mathcal O}$ and hence $T$ and $\ov{T}$ are tensor functors.

The equivalence of categories given in \propref{thm:equivalence} (iii) below is called super duality. The proof of the following proposition is given in \cite[Theorem 6.39]{CW2} and \cite[Theorem 4.6]{CLW} (cf. \cite[Proposition 6.16]{CW2}).
 \begin{prop} \label{thm:equivalence}
\begin{itemize}
\item[(i)] $T:\w{\mc{O}}\rightarrow{\mc{O}}$ is an equivalence of tensor categories.
\item[(ii)] $\ov{T}:\w{\mc{O}}\rightarrow\ov{\mc{O}}$ is an equivalence of tensor
    categories.
\item[(iii)] The tensor categories $\mc{O}$ and $\ov{\mc{O}}$ are equivalent.
\end{itemize}
Moreover, if $ V$ is a highest weight $\DG$-module of highest weight $\w\la$ for $\la\in P^+$,
then $T( V)$ and $\ov T( V)$ are highest weight $\G$--module and $\SG$-module of highest weight $\la$ and $\ov\la$,
respectively. Furthermore, we have
\begin{align*}
T\big{(}\w{\Delta}(\w\la)\big{)}
 =\Delta(\la)\quad \hbox{and}\quad &T\big{(}\w{L}(\w\la)\big{)}=L(\la);
    \\
\ov{T}\big{(}\w{\Delta}(\w\la)\big{)}
 =\ov{\Delta}(\overline\la)\quad \hbox{and}\quad
&\ov{T}\big{(}\w{L}(\w\la)\big{)}=\ov{L}(\overline\la).
\end{align*}
\end{prop}

\section{Super Knizhnik-Zamolodchikov equations}\label{RKZ}
In this section, we first define the Casimir elements (operators) and Casimir symmetric tensors for the Lie (super)algebras $\DG_n$, $\G_n$ and $\SG_n$, and we introduce the (super) Knizhnik-Zamolodchikov equations associated to the Lie superalgebras $\DG_n$, $\G_n$ and $\SG_n$ of finite and infinite ranks. We show that the solutions of the (super) KZ equations are stable under the truncation functors $\w{\mf{tr}}^n_k$, ${\mf{tr}}^n_k$ and $\ov{\mf{tr}}^n_k$. We obtain a bijection between the sets of singular solutions of (super) KZ equations associated to the Lie (super) algebras $\DG_\infty$, $\G_\infty$ and $\SG_\infty$. Finally,  we have a bijection between the sets of solutions of (super) KZ equations
associated to $\w{\mc {G}}_n$ (resp., $\mc {G}_n$ and $\ov{\mc {G}}_n$) and associated to  $\DG_n$ (resp., $\G_n$ and $\SG_n$)
for $n\in\N$.

\subsection{Casimir symmetric tensors} First, we recall the Casimir elements of $\w{\mc {G}}_n$, $\mc {G}_n$ and $\ov{\mc {G}}_n$.
For $n\in \N\cup\{\infty\}$, we define $\w{\varrho}_n^{\xx}$, $\varrho_n^{\xx}$ and $\ov{\varrho}_n^{\xx}$ by
 \begin{align*}
 &\qquad \qquad \qquad \qquad\w{\varrho}_n^{\xx}:=
     \sum_{j\in\w{\I}^+_m(n)}({\bf r}^\xx+\delta_j-(-1)^{2j}j(1-\delta_j))E_j,\\
   & \varrho_n^{\xx}:=
     \sum_{j\in{\I}^+_m(n)}({\bf r}^\xx-j+\delta_j)E_j\,\,
     \quad\hbox{and}\quad
     \ov{\varrho}_n^{\xx}:=
     \sum_{j\in\ov{\I}^+_m(n)}({\bf r}^\xx-(-1)^{2j}(j+\frac{1}{2}\delta_j))E_j
 \end{align*}
for $\xx=\mf{a, b,b^\bullet,c,d}$, where
 \begin{equation*}
   {\bf r}^\xx:=\left\{ \begin{array}{ll}
     -1, & \mbox{for}\ \xx=\mathfrak{a};\\
     -m-\frac{1}{2}, & \mbox{for}\ \xx=\mathfrak{b,b^\bullet};\\
      -m-1, & \mbox{for}\ \xx=\mathfrak{c};\\
      -m, & \mbox{for}\ \xx=\mathfrak{d}
  \end{array} \right.
 \end{equation*}
 and
 \begin{equation*}\label{def:delta}
   \delta_i:=\begin{cases}
     1, & \mbox{for}\ i\ge\hf;\\
    0, & \mbox{for}\ i\in\{-m,\ldots,-1\}.
   \end{cases}
 \end{equation*}
 For $n\in \N$, we have $\w{\varrho}_n^{\xx}\in \w{\mathcal {G}}_n^\xx$, $\varrho_n^{\xx}\in \mathcal {G}_n^\xx$ and $\ov{\varrho}_n^{\xx}\in \ov{\mathcal {G}}_n^\xx$. We also regard $\w{\varrho}_n^{\xx}\in \DG_n^\xx$, $\varrho_n^{\xx}\in \G_n^\xx$ and $\ov{\varrho}_n^{\xx}\in \SG_n^\xx$ by identifying $\w{\mathcal {G}}_n^\xx$ (resp., $\mathcal {G}_n^\xx$ and $\ov{\mathcal {G}}_n^\xx$) as a subspace of the vector space of $\DG_n^\xx$ (resp., $\G_n^\xx$ and $\SG_n^\xx$).
For $n=\infty$, there are only finitely many terms of $\w{\varrho}_n^{\xx}$ (resp., $\varrho_n^{\xx}$ and $\ov{\varrho}_n^{\xx}$) with nonzero actions on each given element in $M$ for $M \in \w{\mathcal {O}}$ (resp., ${\mathcal {O}}$ and $\ov{\mathcal {O}}$) and hence $\w{\varrho}_n^{\xx}$ (resp., $\varrho_n^{\xx}$ and $\ov{\varrho}_n^{\xx}$) is a well defined operator on $M$. As before, we will drop the superscript $\xx$.

For $n\in \N$, the bilinear form $(\cdot, \cdot)$ on $\gl(\wt{V}_m(n))$ defined by
 \begin{equation*}\label{str}
  ( A,B):=  \mbox{Str}(AB)\quad \hbox{for $A,B\in \gl(\wt{V}_m(n))$}
\end{equation*}
is nondegenerate invariant even supersymmetric.
The restriction of the bilinear form $(\cdot, \cdot)$ on $\w{\mathcal{G}}_n$
(resp., ${\mc G}_n$ and $\ov{\mc{G}}_n$), which is also denoted by $(\cdot, \cdot)$,
is also a nondegenerate invariant even supersymmetric bilinear form. To simplify the notations %for the rest of the article,
we assume that
 $$
 \hbox{$\langle\cdot, \cdot\rangle:=(\cdot, \cdot)\,\,\,$ on $\w{\mathcal{G}}^{\mf{a}}_n$, ${\mc G}^{\mf{a}}_n$ and $\ov{\mc{G}}^{\mf{a}}_n$}, \qquad \hbox{and}\qquad  \hbox{$\langle\cdot, \cdot\rangle:=\hf(\cdot, \cdot)\,\,\,$ in other cases.}
  $$
  For each positive root $\beta$ of $\w{\mathcal{G}}_n$
(resp., ${\mc G}_n$ and $\ov{\mc{G}}_n$), we fix two root vectors $E_\beta$ and $E^\beta$ of weights $\beta$ and $-\beta$, respectively, satisfying
 \begin{equation*}
   \langle E_{\beta},E^{\beta}\rangle=1.
 \end{equation*}
Note that $\langle E_i,E_i\rangle=(-1)^{2i}$ for any $i\in\w\I^+_m(n)$.
Then the elements
\begin{align*}
 \mathring{\w{\bf c}}_n&:=2\sum_{\beta\in\w{\Phi}^+_n}E^\beta E_{\beta} +\sum_{j\in\w{\I}^+_m(n)}(-1)^{2j}E_j^2
  +2\w{\varrho}_n,\\
 \mathring{\bf c}_n&:=2\sum_{\beta\in\Phi^+_n}E^\beta E_{\beta} +\sum_{j\in{\I}^+_m(n) }E_j^2
  +2\varrho_n,\\
  \mathring{\ov{\bf c}}_n&:=2\sum_{\beta\in\ov{\Phi}^+_n}E^\beta E_{\beta} +\sum_{j\in\ov{\I}^+_m(n)}(-1)^{2j}E_j^2
  +2\ov{\varrho}_n,
 \end{align*}
which are called the \emph{Casimir elements},
lie in the center of the respective universal enveloping algebras (cf. \cite[Exercise 2.3]{CW2}).
Note that our choice of the pairs of $E_\beta$ and $E^\beta$ are slightly different from the choice of that in \cite{CW2}
and our $\mathring{\w{\bf c}}_n$, and $\mathring{\bf c}_n$ and $\mathring{\ov{\bf c}}_n$
can be obtained from \cite{CW2} by using the fact that
$\langle E_\beta,E^\beta\rangle=(-1)^{|E_\beta|}\langle E^\beta,E_\beta\rangle$. By \eqnref{iso},
the elements
\begin{align*}
 \iota(\mathring{\w{\bf c}}_n)&=2\sum_{\beta\in\w{\Phi}^+_n}E^\beta E_{\beta} +\sum_{j\in\w{\I}^+_m(n)}(-1)^{2j}(E_j-(-1)^{2j}\delta_j K)^2\\
  &\quad +2\sum_{j\in\w{\I}^+_m(n)}({\bf r}+\delta_j-(-1)^{2j}j(1-\delta_j))(E_j-(-1)^{2j}\delta_j K),\\
 \iota(\mathring{\bf c}_n)&=2\sum_{\beta\in\Phi^+_n}E^\beta E_{\beta} +\sum_{j\in{\I}^+_m(n)}(E_j-\delta_j K)^2
  +2\sum_{j\in{\I}^+_m(n)}({\bf r}-j+\delta_j)(E_j-\delta_j K),\\
 \iota( \mathring{\ov{\bf c}}_n)&=2\sum_{\beta\in\ov{\Phi}^+_n}E^\beta E_{\beta} +\sum_{j\in\ov{\I}^+_m(n)}(-1)^{2j}(E_j+\delta_j K)^2
  +2\sum_{j\in\ov{\I}^+_m(n)}({\bf r}-(-1)^{2j}(j+\frac{1}{2}\delta_j))(E_j+\delta_j K)
 \end{align*}
lie in the centers of $U(\DG_n)$, $U(\G_n)$ and $U(\SG_n)$, respectively. Hereafter $U(\mf k)$ stands for the universal enveloping algebra of the Lie (super)algebra $\mf k$.
By removing the terms of $K$ and $K^2$ in the equations above, the elements
 \begin{align*}
 {\w{\bf c}}_n&:=2\sum_{\beta\in\w{\Phi}^+_n}E^\beta E_{\beta} +\sum_{j\in\w{\I}^+_m(n)}(-1)^{2j}(E_j^2-2(-1)^{2j}\delta_j KE_j)
  +2\w{\varrho}_n,\\
 {\bf c}_n&:=2\sum_{\beta\in\Phi^+_n}E^\beta E_{\beta} +\sum_{j\in{\I}^+_m(n)}(E_j^2-2\delta_j KE_j)
  +2{\varrho}_n,\\
 {\ov{\bf c}}_n&:=2\sum_{\beta\in\ov{\Phi}^+_n}E^\beta E_{\beta} +\sum_{j\in\ov{\I}^+_m(n)}(-1)^{2j}(E_j^2+2\delta_j KE_j)
  +2\ov{\varrho}_n,
 \end{align*}
which are also called the \emph{Casimir elements},  also lie in the centers of $U(\DG_n)$, $U(\G_n)$ and $U(\SG_n)$, respectively.

Taking $n\rightarrow \infty$, we define the Casimir operators
 \begin{align*}
 {\w{\bf c}}&:={\w{\bf c}}_\infty :=2\sum_{\beta\in\w{\Phi}^+}E^\beta E_{\beta} +\sum_{j\in\w{\I}^+_m}(-1)^{2j}(E_j^2-2(-1)^{2j}\delta_j KE_j)
  +2\w{\varrho}_\infty,\\
  {\bf c}&:={\bf c}_\infty:=2\sum_{\beta\in\Phi^+}E^\beta E_{\beta} +\sum_{j\in{\I}^+_m}(E_j^2-2\delta_j KE_j)
  +2{\varrho}_\infty,\\
  {\ov{\bf c}}&:={\ov{\bf c}}_\infty:=2\sum_{\beta\in\ov{\Phi}^+}E^\beta E_{\beta} +\sum_{j\in\ov{\I}^+_m}(-1)^{2j}(E_j^2+2\delta_j KE_j)
  +2\ov{\varrho}_\infty.
 \end{align*}

 Similar to the Casimir operators of Kac-Moody algebras, the Casimir operators for $n=\infty$ are well defined operators on modules in the respected categories but they are not elements in the corresponding universal enveloping algebras.

 \begin{lem}\label{finitesum} Let $0\le k<n \leq\infty$. If $v$ is a weight vector of weight $\mu$ in $M \in \w{\mathcal {O}}_n$ (resp., ${\mathcal {O}}_n$ and $\ov{\mathcal {O}}_n$) such that $\mu\in \w\Xi_k$ (resp., $\Xi_k$ and $\ov\Xi_k$), then $E_\beta v= 0$ for all $\beta\in \w{\Phi}_n^{+}\backslash\w{\Phi}_k^{+}$ (resp., ${\Phi}_n^{+}\backslash{\Phi}_k^{+}$ and $\ov{\Phi}_n^{+}\backslash\ov{\Phi}_k^{+}$) and $E_i v= 0$ for $i>k$.

  In particular, for each $v\in M$ and $M \in \w{\mathcal {O}}$, ${\mathcal {O}}$ or $\ov{\mathcal {O}}$, there are only finitely many $E_\beta$ and $E_i$ such that $E_\beta v\not= 0$ and $E_i v\not= 0$.
 \end{lem}

 \begin{proof}  We will prove the case $M \in \w{\mathcal {O}}_n$ only. The proofs of the remaining cases are similar. Let $M \in \w{\mathcal {O}}_n$ and $v\in M$ be a weight vector of weight $\mu$ such that $\mu\in \w\Xi_k$. Therefore $\mu(E_i)=0$ for $i> k$, and hence $E_i v=0$ for $i> k$. For $\beta\in \w{\Phi}_n^{+}\backslash\w{\Phi}_k^{+}$, we have $\beta(E_i)<0$ for some $i> k$ by the description of the positive roots of $\DG$ given in \cite[Sections 6.1.3, 6.1.4]{CW2}. Therefore the weight of $E_\beta v$ does not lie in $\wt{\Xi}$. Hence $E_\beta v=0$.
  \end{proof}

 By \lemref{finitesum} and ${\w{\bf c}}_n$ (resp., ${\bf c}_n$ and ${\ov{\bf c}}_n$) commuting with elements in $\DG_n$ (resp., $\G_n$ and $\SG_n$) for $n\in \N$, ${\w{\bf c}}$ (resp., ${\bf c}$ and ${\ov{\bf c}}$) is a well defined operator on each $M \in \w{\mathcal {O}}$ (resp., ${\mathcal {O}}$ and $\ov{\mathcal {O}}$) and commutes with the action of $\DG$ (resp., $\G$ and $\SG$). They are called the \emph{Casimir operators}.

 For $n\in \N\cup\{\infty\}$, we extend the symmetric bilinear form given in \eqref{b.f. on root system} to a symmetric bilinear form $(\cdot,\cdot)$ on
$\w{\mathfrak{h}}^*_n$ (resp., ${\mathfrak{h}}^*_n$ and $\ov{\mathfrak{h}}^*_n$) defined
by
\begin{equation*}\label{bilinear form}
   (\Lambda_0,\Lambda_0)=0,
   \ \ \ \ (\epsilon_i,\epsilon_j)=(-1)^{2j}\delta_{ij} \ \ \ \ \mbox{and}\ \ \ \  (\Lambda_0,\epsilon_i)=-\delta_i,
 \end{equation*}
 where $i,j\in\w\I^+_m(n)$ (resp., $\I^+_m(n)$ and $\ov\I^+_m(n)$).

For $n\in\N$, it is easy to see the bilinear forms $(\cdot,\cdot)$ on
${\mathfrak{h}}^*_n$ and $\ov{\mathfrak{h}}^*_n$ are non-degenerate.
Therefore there exist unique elements $\rho_n\in{\mf h}^*_n$ and $\ov\rho_n\in\ov{\mf h}^*_n$ satisfying
  \begin{equation*}\label{rho}
(\rho_n,\la)=\la(\varrho_n),\quad \hbox{for $\la\in{\mf h}^*_n$ };\qquad\hbox{and} \qquad
(\ov\rho_n,\la)=\la(\ov\varrho_n),\quad \hbox{for $\la\in\ov{\mf h}^*_n$.}
  \end{equation*}
The bilinear form $(\cdot,\cdot)$ is degenerate on $\w{\mathfrak{h}}^*_n$ for $n\in\N$.
We define $\w\rho_n:=
\sum_{j\in\w{\I}^+_m(n)}({\bf r}^\xx+\delta_j-(-1)^{2j}j(1-\delta_j))\epsilon_j\in\w{\mf h}^*_n$.
Then we have
   \begin{equation*}\label{w-rho}
     (\w\rho_n,\la)=\la(\w\varrho_n),\qquad \hbox{for any $\la\in\w{\mf h}^*_n$.}
   \end{equation*}

  For $n=\infty$, let
 $\w\rho:=\w\rho_\infty$ (resp., $\rho:=\rho_\infty$ and $\ov\rho:=\ov\rho_\infty$)
be an element in the dual space of $\w{\mf h}$ (resp., ${\mf h}$ and $\ov{\mf h}$) defined by
\begin{eqnarray*}
 && \w\rho(E_j):=(-1)^{2j}({\bf r}+\delta_j-(-1)^{2j}j(1-\delta_j)), \ \ \ \ \ \ \ \ \w\rho(K)=0, \\
 && \rho(E_j):={\bf r}+\delta_j-j,\ \ \ \ \ \ \ \ \ \ \ \ \ \ \ \ \ \ \ \ \ \ \ \ \ \ \ \ \ \ \ \ \ \ \ \ \ \ \ \ \ \ \ \,
\rho(K)=0,\\
&& \ov\rho(E_j):=(-1)^{2j}({\bf r}-(-1)^{2j}(j+\frac{1}{2}\delta_j),\ \ \ \ \ \ \ \ \ \ \ \ \ \ \ \, \,  \ov\rho(K)=0.
\end{eqnarray*}
We extend the symmetric bilinear form $(\,,\,)$ on $\w{\mf h}^*$ (resp., ${\mf h}^*$ and $\ov{\mf h}^*$) to the bilinear form $(\,,\,)$ on
$(\w{\mf h}^*\oplus\C\w\rho)\times\w{\mf h}^*$
(resp., $({\mf h}^*\oplus\C\rho)\times{\mf h}^*$ and $(\ov{\mf h}^*\oplus\C\ov\rho)\times\ov{\mf h}^*$),
which is also denoted by $(\,,\,)$, defined by
\begin{align*}
(\w\rho,\epsilon_j)=(-1)^{2j}\w\varrho(E_j)\ \ \ \ \ \ \ \ \mbox{and}\ \ \ \ (\w\rho,\Lambda_0)=0,\qquad \hbox{for $j\in\w\I^+_m$};\\
(\rho,\epsilon_j)=\varrho(E_j)\ \ \ \ \ \ \ \ \mbox{and}\ \ \ \ (\rho,\Lambda_0)=0,\qquad \hbox{for $j\in\I^+_m$};\\
(\ov\rho,\epsilon_j)=(-1)^{2j}\ov\varrho(E_j)\ \ \ \ \ \ \ \ \mbox{and}\ \ \ \ (\ov\rho,\Lambda_0)=0,\qquad \hbox{for $j\in\ov\I^+_m$}.
\end{align*}

For $n\in\N\cup\{\infty\}$ and a weight vector $v$ of weight $\mu$ in $M\in \w{\mathcal O}_n$ (resp., ${\mathcal O}_n$ and $\ov{\mathcal O}_n$), we have
 \begin{equation}\label{rho-v}
 \w\varrho_n v=(\w\rho_n,\mu)v, \qquad \varrho_n v=(\rho_n,\mu)v
 \qquad \hbox{and} \qquad \ov\varrho_n v=(\ov\rho_n,\mu)v.
 \end{equation}

\begin{lem}\label{slem:Gamma}
Let $n\in\N\cup\{\infty\}$. If $v$ is a vector in a highest weight module $V\in \w{\mathcal O}_n$  (resp., ${\mathcal O}_n$ and $\ov{\mathcal O}_n$) of highest weight $\la$,
 then
 \begin{equation*}\label{sGamma2}
  \w{\bf c}_nv=(\la+2\w\rho_n,\la) v\quad
  (\mbox{resp.,}\quad {\bf c}_nv=(\la+2\rho_n,\la) v
\quad\ \mbox{and}
\quad \ov{\bf c}_nv=(\la+2\ov\rho_n,\la) v\,\,).
 \end{equation*}
\end{lem}
\begin{proof}
  We will prove  $ \w{\bf c}_nv=(\la+2\w\rho_n,\la) v$. The others are similar. It is sufficient to assume that $v$ is a highest weight vector in $ V$.
  Let $\la= d\La_0+ \sum_{j\in\w\I^+_m(n)}\la_{j}\epsilon_{j}
 \in \w P^+_n$. By \eqnref{rho-v}, we have
  \begin{eqnarray*}
  \w{\bf c}_n v&=&\big{(}\sum_{j\in\w{\I}^+_m(n)}(-1)^{2j}(E_j^2-2(-1)^{2j}\delta_j KE_j)
  +2\w{\varrho}_n\big{)}v\\
 &=&\sum_{j\in\w{\I}^+_m(n)}((-1)^{2j}\la^2_j -2\delta_j d\la_j)v+2\w{\varrho}_n v\\
 &=&(\la,\la)v+ (2\w\rho_n,\la)v.
 \end{eqnarray*}
 \end{proof}

\begin{prop}\label{thm:scalar} For $\la\in P^+$, we have
\begin{equation*}
 (\w{\lambda}+2\w{\rho},\w{\lambda})= (\lambda+2\rho,\lambda)=(\overline{\lambda}+2\ov{\rho},\overline{\lambda}).
\end{equation*}
\end{prop}
\begin{proof}
We will show $(\w{\lambda}+2\w{\rho},\w{\lambda})= (\lambda+2\rho,\lambda)$.
The proof of the equality
$(\w{\lambda}+2\w{\rho},\w{\lambda})=(\ov{\lambda}+2\ov{\rho},\ov{\lambda})$ is similar.
 By  \cite[Proposition 6.11]{CW2}, the parabolic Verma module $\w\Delta(\w\la)$ of  highest weight $\w \la$ is also a highest weight module of highest weight $\la\in P^+_n$ with respect to $\w{\mathfrak{b}}^c(n)$ for sufficiently large $n$, where $\w{\mathfrak{b}}^c(n)$ is
a Borel subalgebra of $\DG$ corresponding to the Dynkin diagram:
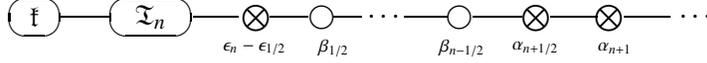
\begin{center}
\hskip -6cm \setlength{\unitlength}{0.16in}
\begin{picture}(24,4)
\put(15.25,2){\makebox(0,0)[c]{$\bigotimes$}}
\put(17.4,2){\makebox(0,0)[c]{$\bigcirc$}}
\put(21.9,2){\makebox(0,0)[c]{$\bigcirc$}}
\put(24.4,2){\makebox(0,0)[c]{$\bigotimes$}}
\put(26.8,2){\makebox(0,0)[c]{$\bigotimes$}}
\put(13.2,2){\line(1,0){1.45}} \put(15.7,2){\line(1,0){1.25}}
\put(17.8,2){\line(1,0){0.9}} \put(20.1,2){\line(1,0){1.4}}
\put(22.35,2){\line(1,0){1.6}} \put(24.9,2){\line(1,0){1.5}}
\put(27.3,2){\line(1,0){1.5}}
\put(19.5,1.95){\makebox(0,0)[c]{$\cdots$}}
\put(29.7,1.95){\makebox(0,0)[c]{$\cdots$}}
\put(15.2,1){\makebox(0,0)[c]{\tiny $\epsilon_n-\epsilon_{1/2}$}}
\put(17.8,1){\makebox(0,0)[c]{\tiny $\beta_{1/2}$}}
\put(22,1){\makebox(0,0)[c]{\tiny $\beta_{n-1/2}$}}
\put(24.4,1){\makebox(0,0)[c]{\tiny $\alpha_{n+1/2}$}}
\put(27,1){\makebox(0,0)[c]{\tiny $\alpha_{n+1}$}}
\put(8.0,2){\makebox(0,0)[c]{{\ovalBox(1.6,1.2){$\mf{k}$}}}}
\put(8.8,2){\line(1,0){1.7}}
\put(11.8,2){\makebox(0,0)[c]{{\ovalBox(2.6,1.2){$\mf{T}_{n}$}}}}
\end{picture}
\end{center}
Let $\mathscr{A}=\{\epsilon_{i-{\hf}}-\epsilon_j|i,j\in\N,\ 1\leq i\leq j\leq n\}$.
Then the set of positive roots of $\DG$ with respect to $\w{\mathfrak{b}}^c(n)$ is
$\w\Phi^{c+}(n):=(\w\Phi^+\backslash\mathscr{A})\cup(-\mathscr{A})$ and
$E^\beta$ is a positive root vector for $\beta\in\mathscr{A}$. Let $v_\la$ be a highest weight vector of $\w\Delta(\w\la)$ with respect to $\w{\mathfrak{b}}^c(n)$ and let $\la= d\La_0+ \sum_{j\in\I^+_m(n)}\la_{j}\epsilon_{j} \in P^+_n$.
 Note that $E^\beta E_\beta=-E_\beta E^\beta-(E_{i-{\hf}}+E_j)$
for $\beta=\epsilon_{i-{\hf}}-\epsilon_j\in\mathscr{A}$. Then
\begin{eqnarray*}
  \w{\bf c}v_\la&=&2\sum_{\beta\in\mathscr{A}}E^\beta E_{\beta}v_\la +\sum_{j\in\w{\I}^+_m}(-1)^{2j}(E_j^2-2(-1)^{2j}\delta_j KE_j)v_\la   +2\w{\varrho}v_\la\\
  &=&
  -2\sum_{\beta\in\mathscr{A}}E_{\beta}E^\beta v_\la -2\sum_{1\leq i\leq j\leq n}(E_{i-{\hf}}+E_j) v_\la
  +\sum_{j\in{\I}^+_m}(\la_j^2-2\delta_j d\la_j)v_\la+2\w{\varrho}v_\la
   \\
    &=&
  -2\sum_{i=1}^n i\la_i v_\la+(\la,\la)v_\la
  +2\w{\varrho}v_\la.
\end{eqnarray*}
From the definitions of $\w{\varrho}$ and ${\varrho}$, we have $\w{\varrho}v_\la -\sum_{i=1}^n i\la_i v_\la={\varrho}v_\la$. Therefore  $\w{\bf c}v_\la=(\la +2 \rho, \la)v_\la$. On the other hand, $\w{\bf c}v_\la=(\w\la+2\w\rho,\w\la)v_\la$  by \lemref{slem:Gamma}. Hence $(\w\la+2\w\rho,\w\la)= (\la +2  \rho, \la)$.
\end{proof}

\begin{rem}
     A combinational proof of  the equality $({\lambda}+2{\rho},{\lambda})=(\ov{\lambda}+2\ov{\rho},\ov{\lambda})$ for type $\mathfrak{a}$ is given in \cite[Lemma 3.3]{CKL}.
\end{rem}

 Now, for $n\in \N\cup\{\infty\}$, we define
  \begin{align*}
 {\w{\Omega}}_n&:=\sum_{\beta\in\w{\Phi}^+_n}(E^\beta\otimes E_{\beta}+(-1)^{|E_\beta|} E_\beta \otimes E^{\beta})  +\sum_{j\in\w{\I}^+_m(n)}(-1)^{2j}\big(E_j\otimes E_j-(-1)^{2j}\delta_j  (K\otimes E_j+E_j\otimes K)\big),\\
 {\Omega}_n&:=\sum_{\beta\in\Phi^+_n} \big(E^\beta\otimes E_{\beta}+(-1)^{|E_\beta|} E_\beta\otimes E^{\beta}) +\sum_{j\in{\I}^+_m(n)}(E_j\otimes E_j-\delta_j (K\otimes E_j+E_j\otimes K)\big),\\
 {\ov{\Omega}}_n&:=\sum_{\beta\in\ov{\Phi}^+_n}(E^\beta\otimes E_{\beta}+(-1)^{|E_\beta|} E_\beta \otimes E^{\beta})  +\sum_{j\in\ov{\I}^+_m(n)}(-1)^{2j}\big(E_j\otimes E_j+\delta_j  (K\otimes E_j+E_j\otimes K)\big),
  \end{align*}
 which are called \emph{Casimir symmetric tensors}.
As before, we will drop the subscript for $n=\infty$. Hereafter $\Delta$ denotes the comultiplication on a universal enveloping algebra.  For $n\in \N$, it is easy to see that ${\w{\Omega}}_n$, ${{\Omega}}_n$ and  ${\ov{\Omega}}_n$ are elements in $U(\DG_n)\otimes U(\DG_n)$, $U(\G_n)\otimes U(\G_n)$ and $U(\SG_n)\otimes U(\SG_n)$, respectively, satisfying the following equations:
  \begin{align}\label{Oc}
  \nonumber \w{\Omega}_n&=\hf\big(\Delta({\w{\bf c}}_n)- {\w{\bf c}}_n\otimes 1 -1\otimes {\w{\bf c}}_n\big), \\
  \Omega_n&=\hf(\Delta({\bf c}_n)-{\bf c}_n\otimes 1-1\otimes {\bf c}_n),\\
  \nonumber  {\ov{\Omega}}_n&=\hf(\Delta({\ov{\bf c}}_n)-{\ov{\bf c}}_n\otimes 1-1\otimes {\ov{\bf c}}_n).
\end{align}
It is easy to see that these equations also hold for $n=\infty$ by regarding both sides of the equations acting on $M\otimes N$ for $M, N\in \w{\mathcal {O}}$ (resp., ${\mathcal {O}}$ and $\ov{\mathcal {O}}$) in the following sense. $\w{\Omega}$ (resp., ${\Omega}$ and $\ov{\Omega}$) is regarded as a well defined operator on $M\otimes N$ by \lemref{finitesum} and $\Delta(\w{\bf c})$  (resp., $\Delta({\bf c})$ and $\Delta(\ov{\bf c})$) is regarded as the action of $\w{\bf c}$  (resp., ${\bf c}$ and $\ov{\bf c}$) on $M\otimes N$.
The following proposition is a direct consequence of equations obtained in \eqnref{Oc}.

\begin{prop}\label{Om-g}
 For $n\in \N\cup\{\infty\}$ and $M, N\in \w{\mathcal {O}}_n$ (resp., ${\mathcal {O}}_n$ and $\ov{\mathcal {O}}_n$), the action of $\w{\Omega}_n$ (resp., ${\Omega}_n$ and $\ov{\Omega}_n$) on $M\otimes N$ commutes with the action of $\DG_n$ (resp., $\G_n$ and $\SG_n$).
\end{prop}

\subsection{(Super) Knizhnik-Zamolodchikov equations}\label{SKZ}

Let $n\in\mathbb{N}\cup\{\infty\}$ and
let ${M}_i\in \w\OO_n$ for  $i=1,\ldots,\ell$ and let
 \begin{equation}\label{def:M}
 M:= M_1\otimes\cdots\otimes  M_\ell.
\end{equation}
 Then ${M}\in \w\OO_n$.
For $x\in\w{\mathfrak{g}}_n$, we write $x^{(i)}=\underbrace{1\otimes\cdots\otimes1\otimes \stackrel{i}{x}\otimes1\otimes\cdots\otimes1}_{\ell}$, acting on  $M$. Given any operator $A=\sum_{r\in I} x_r\otimes y_r$ on $N_1\otimes N_2$ with $x_r, y_r\in\w{\mathfrak{g}}_n$ for any $N_1, N_2\in \w\OO_n$, let $A^{(ij)}:=\sum_{r\in I}x_r^{(i)}y_r^{(j)}$ for $1\leq i,j\leq\ell$. Then $A^{(ij)}$ is an operator on $M$  for $i\not=j$.
Recall that the Casimir symmetric tensor $\w{\Omega}_n$ is an operator on $N_1\otimes N_2$ for any $N_1, N_2\in \w\OO_n$ and hence $\w{\Omega}^{(ij)}_n$ is an operator on $M$.
Similarly, we can define the operators ${\Omega}^{(ij)}_n$ (resp., $\ov{\Omega}^{(ij)}_n$) on the tensor product \eqref{def:M} for $M_1,\ldots,M_\ell\in\mathcal{O}_n$ (resp., $\ov{\mathcal {O}}_n$) of modules for $i\neq j$.

The the following lemma follows from the fact that $x^{(i)}y^{(j)}=(-1)^{|x||y|}y^{(j)}x^{(i)}$ for any homogeneous elements $x$ and $y$ in a Lie superalgebra and $1\le i<j \le \ell$.

\begin{lem}\label{ij=ji}
Let $1\le i<j \le \ell$. We have  $\w{\Omega}^{(ji)}_n=\w{\Omega}^{(ij)}_n$,
  ${\Omega}^{(ji)}_n={\Omega}^{(ij)}_n$ and
$\ov{\Omega}^{(ji)}_n=\ov{\Omega}^{(ij)}_n$.
\end{lem}

Let ${\bf X}_\ell:=\{(z_1, \ldots,z_\ell)\in \C^\ell\,|\, z_i\not=z_j, \,\, \hbox{for any $i\not=j$}\}$ denote the \emph{configuration space} of $\ell$ distinct points in $\C$. For a nonempty open subset $U$ in ${\bf X}_\ell$ and a finite-dimensional vector space $N$,  let $\mc{D}(U, N)$ denote the set of differentiable functions from $U$ to $N$. For $N\in \w{\mathcal {O}}_n$ (resp., ${\mathcal {O}}_n$  and $\ov{\mathcal {O}}_n$), let
 \[
 \mc{D}(U, N):=\bigoplus_\mu \mc{D}(U, N_\mu),
 \]
 where $\mu$ runs over all weights of $N$. It is clear that $\mc{D}(U, N)$ is a $\DG_n$- (resp., $\G_n$- and $\SG_n$-)module for $N\in \w{\mathcal {O}}_n$ (resp., ${\mathcal {O}}_n$ and $\ov{\mathcal {O}}_n$).

The \emph{(quadratic) Gaudin Hamiltonians} ${\w{\mathcal {H}}}^i_{n}$, for $i=1,\ldots, \ell$, are linear operators  on  $\mc{D}(U, M)$ defined by
\begin{equation*}\label{GH_n}
  \w{\mathcal {H}}^i_{n}:=\sum_{j=1\atop j\neq i}^\ell\frac{ \w{\Omega}^{(ij)}_n}{z_i-z_j}.
\end{equation*}
Similarly, we can define the \emph{(quadratic) Gaudin Hamiltonians} ${{\mathcal {H}}}^i_{n}$ and ${\ov{\mathcal {H}}}^i_{n}$, for $i=1, \ldots, \ell$, being linear operators  on  $\mc{D}(U, M)$ by
\begin{equation*}
  {\mathcal {H}}^i_{n}:=\sum_{j=1\atop j\neq i}^\ell\frac{{\Omega}^{(ij)}_n}{z_i-z_j}\quad \hbox{and}\quad \ov{\mathcal {H}}^i_{n}:=\sum_{j=1\atop j\neq i}^\ell\frac{ \ov{\Omega}^{(ij)}_n}{z_i-z_j}
\end{equation*}
if  ${M}_1,\ldots, M_\ell\in \OO_n$ and ${M}_1, \ldots, M_\ell\in \ov\OO_n$, respectively.
Fix a nonzero complex number $\kappa$ and ${\psi}(z_1,\ldots, z_\ell)\in \mc{D}(U, M)$. We can consider a system of partial differential equations
\begin{equation}\label{dKZE_n}
\kappa\frac{\partial}{\partial {z_i}}{\psi}(z_1,\ldots,z_\ell)=\w{\mathcal {H}}^i_{n}{\psi}(z_1,\ldots,z_\ell),\ \ \ \ \mbox{for}\ i=1,\ldots,\ell.
\end{equation}
Above equations \eqref{dKZE_n} are called the \emph{super Knizhnik-Zamolodchikov equations} (super KZ equations for short).
Similarly, we can consider the \emph{KZ equations} and \emph{super KZ equations}
\begin{equation}\label{KZE_n}
\kappa\frac{\partial}{\partial {z_i}}{\psi}(z_1,\ldots,z_\ell)={\mathcal {H}}^i_{n}{\psi}(z_1,\ldots,z_\ell),\quad  \mbox{for}\ i=1,\ldots,\ell, \quad {\psi}(z_1,\ldots, z_\ell)\in \mc{D}(U, M)
\end{equation}
and
\begin{equation}\label{sKZE_n}
\kappa\frac{\partial}{\partial {z_i}}{\psi}(z_1,\ldots,z_\ell)=\ov{\mathcal {H}}^i_{n}{\psi}(z_1,\ldots,z_\ell),\quad  \mbox{for}\ i=1,\ldots,\ell, \quad {\psi}(z_1,\ldots, z_\ell)\in \mc{D}(U, M)
\end{equation}
for ${M}_1, \ldots, M_\ell\in \OO_n$ and ${M}_1, \ldots, M_\ell\in \ov\OO_n$, respectively.
For ${M}_1,\ldots, M_\ell\in \w\OO_n$ and a weight  $\mu$ of $M$, let
 \[\w{\rm KZ}(M_\mu):=\{\psi\in \mc{D}(U, M_\mu)\,|\, \hbox{$\psi$ is a solution of the super KZ equations \eqnref{dKZE_n}}\}
 \]
 and
\[\w{\rm KZ}(M):=\bigoplus_\mu \w{\rm KZ}(M_\mu)\]
where $\mu$ runs over all weights of $M$.
Analogously, we let ${\rm KZ}(M_\mu)$ and $ \ov{\rm KZ}(M_\mu)$ denote the set of all solutions of  \eqnref{KZE_n} and  \eqnref{sKZE_n}
in $\mc{D}(U, M_\mu)$, respectively.
Also,  let
 \[
 {\rm KZ}(M):=\bigoplus_{\mu} {\rm KZ}(M_\mu),\quad\hbox{and} \quad \ov{\rm KZ}(M):=\bigoplus_{\mu}  \ov{\rm KZ}(M_\mu),
 \]
where $\mu$ runs over all weights of $M$.
Let $\w{\mc S}(\w M_\mu)$ (resp., ${\mc S}(M_{\mu})$ and $\ov{\mc S}( M_\mu)$)
 be the set of functions in
 $\psi\in \w{\rm KZ}(M_\mu)$ (resp., ${\rm KZ}(M_\mu)$ and $\ov{\rm KZ}(M_\mu)$)
 satisfying $E_\beta \psi=0$ for $\beta\in\w{\Phi}_n^+$ (resp., ${\Phi}_n^+$ and $\ov{\Phi}_n^+$).
 The sets are composed of the solutions of (super) KZ equations with values in the subspaces spanned by singular vectors of weight $\mu$. $\w{\mc S}(\w M_\mu)$ (resp., ${\mc S}(M_{\mu})$ and $\ov{\mc S}( M_\mu)$,
 which is called the \emph{singular solution space} of the (super) KZ equations, is a vector space.

The proof of the following proposition is standard. We include the proof for completeness.

\begin{prop}\label{H-g}
  The Gaudin Hamiltonians
  $\w{\mathcal {H}}^i_{n}$, (resp., ${\mathcal {H}}^i_{n}$ and
  $\ov{\mathcal {H}}^i_{n}$) on $\mc{D}(U, M)$ mutually commute with each other, and they also commute with the action of Lie (super)algebra $\wt{\G}_n$ (resp., $\G_n$ and $\SG_n$), for $i=1,\ldots,\ell$.
\end{prop}
\begin{proof}
The second statement follows directly from \propref{Om-g}. We will show $[\w{\mathcal {H}}^i_{n}, \w{\mathcal {H}}^j_{n}]=0$ for $1\le i,j\le\ell$. The other cases are similar.

Note that $[\w\Omega^{(ij)}_n, b^{(p)}]=0\ \mbox{for}\ p\neq i, j$ and $[\w\Omega^{(ij)}_n, b^{(i)}+b^{(j)}]=[\w\Omega^{(ij)}_n, \Delta(b)^{(ij)}]=0$
for all $b\in\wt{\G}_n$, by \propref{Om-g}. By writing
$\w\Omega^{(pi)}_n=\sum\limits_{c,d}c^{(p)}d^{(i)}$ and
$\w\Omega^{(pj)}_n=\sum\limits_{c,d}c^{(p)}d^{(j)}$ with $c,d\in\DG_n$ and $|c|=|d|$,  we have
\begin{equation}\label{[ij,pi+pj]}
  [\w\Omega^{(ij)}_n,\w\Omega^{(pi)}_n+\w\Omega^{(pj)}_n]=0\ \  \mbox{for  $p\neq i, j$ and $i\not=j$}.
\end{equation}

By  \lemref{ij=ji} and \eqref{[ij,pi+pj]}, we have, for $i\not=p$,
\begin{eqnarray*}
[\w{\mathcal {H}}^i_{n}, \w{\mathcal {H}}^p_{n}]
  &=&\sum_{j=1\atop j\neq i, p}^\ell[\frac{\w\Omega^{(ij)}_n}{z_i-z_j},\frac{\w\Omega^{(pi)}_n}{z_p-z_i}+\frac{\w\Omega^{(pj)}_n}{ z_p-z_j}]
  +\sum_{j=1\atop j\neq i, p}^\ell[\frac{\w\Omega^{(ip)}_n}{ z_i-z_p},\frac{\w\Omega^{(pj)}_n}{ z_p-z_j}]\\
  &=&\sum_{j=1\atop j\neq i, p}^\ell\frac{[\w\Omega^{(ij)}_n,\w\Omega^{(pi)}_n]}{(z_i-z_j)(z_p-z_i)}
  +\sum_{j=1\atop j\neq i, p}^\ell \frac{[\w\Omega^{(ij)}_n,\w\Omega^{(pi)}_n+\w\Omega^{(pj)}_n-\w\Omega^{(pi)}_n]}{(z_i-z_j)(z_p-z_j)}\\
 &&\quad+ \sum_{j=1\atop j\neq i, p}^\ell\frac{[\w\Omega^{(ip)}_n,\w\Omega^{(pj)}_n+\w\Omega^{(ij)}_n-\w\Omega^{(ij)}_n]}{(z_i-z_p)(z_p-z_j)}
  \\
  &=&\sum_{j=1\atop j\neq i, p}^\ell[\w\Omega^{(ij)}_n,\w\Omega^{(pi)}_n]
  \big(\frac{1}{  (z_i-z_j)(z_p-z_i)}+\frac{1}{  (z_i-z_j)(z_j-z_p)}+\frac{1}{ (z_p-z_i)(z_j-z_p)}\big)\\
  &=&0.
\end{eqnarray*}
 \end{proof}

The following proposition is a direct consequence of the second part of Proposition \ref{H-g}.
\begin{prop}\label{KZ-g}
  $\w{\rm KZ}(M)$ (resp., ${\rm KZ}(M)$ and  $\ov{\rm KZ}(M)$) is a $\DG_n$-(resp., $\G_n$- and $\SG_n$-)module.
\end{prop}

 From the definition of
 $\w{\mf{tr}}^n_k$ (resp., ${\mf{tr}}^n_k$ and $\ov{\mf{tr}}^n_k$),
 we can define the function
 $\w{\mf{tr}}^n_k({\psi})$
  (resp., ${\mf{tr}}^n_k({\psi})$ and $\ov{\mf{tr}}^n_k({\psi})$)
  in $\mc{D}(U,\w{\mf{tr}}^n_k(M)_\mu)$
  (resp., $\mc{D}(U,{\mf{tr}}^n_k(M)_\mu)$ and $\mc{D}(U,\ov{\mf{tr}}^n_k(M)_\mu)$)
   in an obvious way for
   ${\psi}(z_1,\ldots,z_\ell)\in\mc{D}(U,M_\mu)$,
 where $\mu$ is a weight of $M$ for $M=M_1\otimes\cdots\otimes  M_\ell$ and
  $M_i\in \w{\mc {O}}_n$ (resp., ${\mc {O}}_n$ and $\ov{\mc {O}}_n$) for $i=1,\ldots,\ell$.
  Note that $\w{\mf{tr}}^n_k({\psi})$
  (resp., ${\mf{tr}}^n_k({\psi})$ and $\ov{\mf{tr}}^n_k({\psi})$) $=\psi$,
  if $\mu\in \w\Xi_k$ (resp., $\Xi_k$ and $\ov\Xi_k$), and $0$, otherwise.

\begin{prop}\label{tr-KZ} Let $0\le k<n \leq\infty$ and $M=M_1\otimes\cdots\otimes  M_\ell$. Assume that  ${\psi}\in \mc{D}(U, M)$.
\begin{itemize}
\item[(i)] For $\mu\in \w\Xi_k$ and $M_1, \ldots,M_\ell\in \w{\mc {O}}_n$, we have
\[
\quad
{\psi}\in \w{\rm KZ}(M_\mu) \Leftrightarrow \w{\mf{tr}}^n_k({\psi})\in \w{\rm KZ}\big(\w{\mf{tr}}^n_k(M)_\mu\big)
\quad\hbox{and}\quad
{\psi}\in \w{\mc S}(M_\mu) \Leftrightarrow \w{\mf{tr}}^n_k({\psi})\in \w{\mc S}\big(\w{\mf{tr}}^n_k(M)_\mu\big).
\]
\item[(ii)] For $\mu\in \Xi_k$ and $M_1, \ldots,M_\ell\in{\mc {O}}_n$, we have
\[
\quad
{\psi}\in {\rm KZ}(M_\mu) \Leftrightarrow {\mf{tr}}^n_k({\psi})\in {\rm KZ}\big({\mf{tr}}^n_k(M)_\mu\big)
\quad\hbox{and}\quad
{\psi}\in {\mc S}(M_\mu) \Leftrightarrow {\mf{tr}}^n_k({\psi})\in {\mc S}\big({\mf{tr}}^n_k(M)_\mu\big).
\]
\item[(iii)] For $\mu\in \ov\Xi_k$ and $M_1, \ldots,M_\ell\in\ov{\mc {O}}_n$, we have
\[
\quad
{\psi}\in \ov{\rm KZ}(M_\mu) \Leftrightarrow \ov{\mf{tr}}^n_k({\psi})\in \ov{\rm KZ}\big(\ov{\mf{tr}}^n_k(M)_\mu\big)
\quad\hbox{and}\quad
{\psi}\in \ov{\mc S}(M_\mu) \Leftrightarrow \ov{\mf{tr}}^n_k({\psi})\in \ov{\mc S}\big(\ov{\mf{tr}}^n_k(M)_\mu\big).
\]
\end{itemize}
\end{prop}

\begin{proof}
We will show (i) only. The proofs of (ii) and (iii) are similar.
 Write $\psi=\sum_{r=1}^p\psi_r$ such that
 $\psi_r(z_1,\ldots,z_\ell)\in (M_1)_{\mu^{r,1}}\otimes\cdots\otimes(M_\ell)_{\mu^{r,\ell}}$ for all $(z_1,\ldots,z_\ell)\in U$,
 where $\mu^{r,s}$ are weights of $M_s$ for $s=1,\ldots, \ell$ and $r=1,\ldots, p$.
 Note that $\sum_{s=1}^\ell \mu^{r,s}=\mu$  for $r=1,\ldots, p$.
  By \lemref{weight=0}, $\mu^{r,s}\in  \w\Xi_k$, for all $s=1,\ldots, \ell$ and $r=1,\ldots, p$.
  Therefore $\w{\mf{tr}}^n_k({\psi})=\psi=\sum_{r=1}^p\psi_r$. By \lemref{finitesum}, $\w\Omega_n^{(ij)}\psi=\w\Omega_k^{(ij)}\psi=\w\Omega_k^{(ij)}\w{\mf{tr}}^n_k({\psi})$
  for $1\le i\not= j\le \ell$ since $\mu^{r,i}, \mu^{r,j}\in  \w\Xi_k$,
  for $r=1,\ldots, p$.
  Hence ${\psi}\in \w{\rm KZ}(M_\mu)$ if and only if $\w{\mf{tr}}^n_k({\psi})\in \w{\rm KZ}\big(\w{\mf{tr}}^n_k(M)_\mu\big)$.
  The second part follows from the first part
  and the fact that $v$ is a singular vector of weight $\mu\in \w\Xi_k$ in $M$ if and only if $E_\beta v=0$ for all $\beta\in \w\Phi_k^+$.
\end{proof}

\subsection{KZ equations for modules over $\DG$, $\G$ and $\SG$.}\label{Sec: KZ_infty}
Let $\w{M}_i\in \w\OO$ for  $i=1,\ldots,\ell$ and let
 $$
 \w{M}:= \w{M}_1\otimes\cdots\otimes \w{M}_\ell.
 $$
We set $M_i:=T(\w{M}_i)$ and $\ov M_i:=\ov{T}(\w{M}_i)$ for  $i=1,\ldots,\ell$.
Then
 \[
 T(\w{M})={M}_1\otimes\cdots\otimes {M}_\ell\quad \hbox{and}\quad \ov{T}(\w{M})= \ov{M}_1\otimes\cdots\otimes \ov{M}_\ell.
 \]
 Recall that $T$ and  $\ov T$ are tensor functors defined in \eqref{def:T}.
  We define $T{\psi}(z_1,\ldots,z_\ell)\in \mc{D}(U, T(\w{M}))$
  and  $\ov T{\psi}(z_1,\ldots,z_\ell)\in \mc{D}(U, \ov T(\w{M}))$
  for ${\psi}(z_1,\ldots,z_\ell)\in \mc{D}(U, \w{M})$
  in an obvious way by letting
  $T{\psi}:=\psi$ (resp., $\w T{\psi}:=\psi$), for ${\psi}\in \mc{D}(U, \w{M}_\mu)$, if $\mu\in \Xi$
  (resp., $\mu\in \ov\Xi$), and $0$, otherwise.
  Recall that we drop the $\infty$ for $n=\infty$.

 \begin{lem}\label{Xi} Let $N \in \w{\mathcal {O}}$ and let $v\in N$ be a weight vector of weight $\mu$.
\begin{itemize}
\item[(i)] For $\mu\in \Xi$, we have $E_i v= 0$ for all $i\in \w{\I}^+_m\backslash {\I}^+_m$ and either $E_\beta v= 0$ or $E^\beta v= 0$ for $\beta\in \w{\Phi}^{+}\backslash {\Phi}^{+}$.
\item[(ii)] For $\mu\in \ov\Xi$, we have $E_i v= 0$ for all $i\in \w{\I}^+_m\backslash \ov{\I}^+_m$ and either $E_\beta v= 0$ or $E^\beta v= 0$ for  $\beta\in \w{\Phi}^{+}\backslash\ov{\Phi}^{+}$.
\end{itemize}
  \end{lem}

 \begin{proof}  We will show (i). The proof of (ii) is similar.  Let  $v\in N$ be a weight vector of weight $\mu$ such that $\mu\in \Xi$. Then $\mu(E_i)=0$ for all $i\in \w{\I}^+_m\backslash {\I}^+_m$  and hence $E_i v=0$ for all $i\in \w{\I}^+_m\backslash {\I}^+_m$. For $\beta\in \w{\Phi}^{+}\backslash {\Phi}^{+}$, we have $\beta(E_i)\not= 0$ for some $i\in \w{\I}^+_m\backslash {\I}^+_m$. Therefore either the weight of $E_\beta v$ or $E^\beta v$ does not lie in $\wt{\Xi}$. Hence either $E_\beta v=0$ or $E^\beta v=0$.
  \end{proof}

   \begin{lem}\label{wOm-Om} Let $M, N \in \w{\mathcal {O}}$ and let $v\in M\otimes N$ be a weight vector of weight $\mu$.
\begin{itemize}
\item[(i)] For $\mu\in \Xi$, we have  $\w\Omega v=\Omega v$.
\item[(ii)] For $\mu\in \ov\Xi$, we have $\w\Omega v=\ov\Omega v$.
\end{itemize}
  \end{lem}

  \begin{proof}  We will show (i). The proof of (ii) is similar.  We may assume that  the weight vector $v=v_1\otimes v_2$ such that $v_1\in M$ and $ v_2\in N$ are weight vectors of weight $\mu^1$ and $\mu^2$, respectively. We have  $\mu=\mu^1+ \mu^2\in \Xi$. By \lemref{weight=0}, we have $\mu^1, \mu^2\in  \Xi$. For $\beta\in \w{\Phi}^{+}\backslash {\Phi}^{+}$, we have $\beta(E_i)\not= 0$ for some $i\in \w{\I}^+_m\backslash {\I}^+_m$. Therefore either the weight of $E_\beta v_1$ or $E^\beta v_2$ does not lie in $\wt{\Xi}$. Thus either $E_\beta v_1=0$ or $E^\beta v_2=0$ and hence $E_\beta\otimes E^\beta (v_1\otimes v_2)=0$ for $\beta\in \w{\Phi}^{+}\backslash {\Phi}^{+}$. Similarly, $E^\beta\otimes E_\beta (v_1\otimes v_2)=0$ for $\beta\in \w{\Phi}^{+}\backslash {\Phi}^{+}$. By \lemref{Xi}, we have $E_i v_1=0$ and $E_i v_2=0$ for all $i\in \w{\I}^+_m\backslash {\I}^+_m$. Therefore $\w\Omega(v_1\otimes v_2)=\Omega (v_1\otimes v_2)$.
  \end{proof}

\begin{thm}\label{T-KZ} Let ${\psi}\in \mc{D}(U, \w M)$. We have
\begin{itemize}
\item[(i)] For $\mu\in \Xi$, ${\psi}\in \w{\rm KZ}(\w{M}_\mu)$ if and only if $T{\psi}\in {\rm KZ}(T(\w{M})_\mu)$,
\item[(ii)] For $\mu\in \ov\Xi$, ${\psi}\in \w{\rm KZ}(\w{M}_\mu)$ if and only if $\ov T{\psi}\in {\rm KZ}(\ov T(\w{M})_\mu)$.
\end{itemize}
\end{thm}

\begin{proof} We will show (i). The proof of (ii) is similar.
Write $\psi=\sum_{r=1}^p\psi_r$ such that
$\psi_r(z_1,\ldots,z_\ell)\in (M_1)_{\mu^{r,1}}\otimes\cdots\otimes(M_\ell)_{\mu^{r,\ell}}$
for all $(z_1,\ldots,z_\ell)\in U$,
where $\mu^{r,s}$ are weights of $M_s$ for $s=1,\ldots, \ell$ and $r=1,\ldots, p$.
Note that $\sum_{s=1}^\ell \mu^{r,s}=\mu$  for $r=1,\ldots, p$.
By \lemref{weight=0}, we have $\mu^{r,s}\in  \Xi$, for all $s=1,\ldots, \ell$ and $r=1,\ldots, p$.
Therefore $T{\psi}=\psi=\sum_{r=1}^p\psi_r$.
By \lemref{wOm-Om}, $\w\Omega^{(ij)}\psi=\Omega^{(ij)}\psi=\Omega^{(ij)}T{\psi}$
for $1\le i\not= j\le \ell$
since $\mu^{r,i}, \mu^{r,j}\in  \Xi$, for $r=1,\ldots, p$.
Hence ${\psi}\in \w{\rm KZ}(M_\mu)$ if and only if $T\psi\in \w{\rm KZ}\big(T(\w M)_\mu\big)$.
\end{proof}

The following lemma follows from the proof of \cite[Theorem 4.6]{CLW} (see also \cite[Lemma 3.2]{CL}).

\begin{lem}
\label{matching:weights}
Let $\la\in P^+$ and $V=\w\Delta(\w\la)$ or $\w L(\w\la)$. Then ${\rm dim} V_{\w\la}={\rm dim} V_{\la}={\rm dim} V_{\ov\la}=1$ and there are $\emph{X}_\la, \emph{Y}_\la, \ov{\emph{X}}_\la, \ov{\emph{Y}}_\la\in U(\w{\mf{l}})$ such that
\begin{equation*}
\begin{aligned}[c]
\varphi : \, V_{\w\la}& \longrightarrow V_{\la}\\
 v & \mapsto \emph{Y}_\la v
\end{aligned}
\qquad\hbox{and} \qquad
\begin{aligned}[c]
\phi : \, V_{\w\la}& \longrightarrow V_{\ov\la}\\
 v & \mapsto \ov{\emph{Y}}_\la v
\end{aligned}
\end{equation*}
are linear isomorphisms with $\varphi^{-1}(u)=\emph{X}_\la u$ for $u\in V_{\la}$ and $\phi^{-1}(w)=\ov{\emph{X}}_\la w$ for $w\in V_{\ov\la}$.
\end{lem}

The following theorem is a consequence of \thmref{T-KZ}, \lemref{matching:weights} and the fact that $\w{\rm KZ}(\w{M})$ is a $\DG$-module.

\begin{thm}\label{T-KZ-odd} Let $\la\in P^+$.
 \begin{itemize}
     \item[(i)] If ${\psi}\in  \w{\rm KZ}(\w{M}_{\w\la})$, then  $\emph{Y}_\la{\psi}\in  {\rm KZ}(T(\w M)_\la)$ and $\ov{\emph{Y}}_\la{\psi}\in \ov{\rm KZ}(\ov T(\w M)_{\ov\la})$.
     \item[(ii)] If ${\psi}\in \mc{D}(U, \w M_{\la})$ and $T{\psi}\in {\rm KZ}(T(\w M)_\la)$, then $\emph{X}_\la{\psi}\in \w{\rm KZ}(\w{M}_{\w\la})$.
      \item[(iii)] If ${\psi}\in \mc{D}(U, \w M_{\ov \la})$ and $\ov T {\psi}\in \ov{\rm KZ}(\ov T({\w M})_{\ov\la})$, then $\ov{\emph{X}}_\la{\psi}\in \w{\rm KZ}(\w{M}_{\w\la})$.
   \end{itemize}
 \end{thm}

Recall $\w{\mc S}(\w{M}_{\w\la})$ (resp., ${\mc S}(T(\w M)_{\la})$ and $\ov{\mc S}(\ov T(\w M)_{\ov\la})$,  for $\la\in P^+$,
denotes the set of functions $\psi\in \w{\rm KZ}(\w{M}_{\w\la})$ (resp., ${\rm KZ}(T(\w M)_\la)$ and $\ov{\rm KZ}(\ov T(\w M)_{\ov\la})$) satisfying $E_\beta \psi=0$ for $\beta\in\w{\Phi}^+$ (resp., ${\Phi}^+$ and $\ov{\Phi}^+$).

 The following theorem follows from \thmref{T-KZ-odd}, \lemref{matching:weights} and a consequence of \propref{thm:equivalence} that there are natural isomorphisms ${\rm Hom}_{\w{\mc \OO}}(\w\Delta(\w\la), \w M)\cong {\rm Hom}_{{\mc \OO}}(\Delta(\la),  T(\w M))$ and ${\rm Hom}_{\w{\mc \OO}}(\w\Delta(\w\la), \w M)\cong {\rm Hom}_{\ov{\mc \OO}}(\ov\Delta(\ov\la), \ov T(\w M))$ for $\la\in P^+$.
 \begin{thm}
\label{sing-sol}
Let $\la\in P^+$. There are linear isomorphisms
\begin{equation*}
\begin{aligned}[c]
\varphi : \, \w{\mc S}(\w{M}_{\w\la}) & \longrightarrow {\mc S}(T(\w M)_{\la})\\
\psi & \mapsto \emph{Y}_\la \psi
\end{aligned}
\qquad\hbox{and} \qquad
\begin{aligned}[c]
 \qquad\phi :\, \w{\mc S}(\w{M}_{\w\la})& \longrightarrow \ov{\mc S}(\ov T(\w M)_{\ov\la})\\
 \psi & \mapsto \ov{\emph{Y}}_\la \psi
\end{aligned}
\end{equation*}
with $\varphi^{-1}(\psi )=\emph{X}_\la \psi $  for $\psi\in{\mc S}(T(\w M)_{\la})$ and $\phi^{-1}(\ov \psi )=\ov{\emph{X}}_\la \ov \psi $ for $\ov \psi\in\ov{\mc S}(\ov T(\w M)_{\ov\la})$.
\end{thm}

 The following lemma follows from the description of the positive root system of $\DG_n$ given in \cite[Sections 6.1.3 and 6.1.4]{CW2} and possible weights of a module given in \eqnref{weight}.

 \begin{lem}\label{k_0}
Let $n\in \N$ and $\la^i\in P^+$ such that  $\ov\la^i\in\ov P^+_n$, for $i=1,\ldots,\ell$.
Let $\ov V_i$ denote the parabolic Verma module $\ov\Delta_n(\ov\la^i)$  in $\ov\OO_n$ of highest weight $\ov\la^i$, for $i=1,\ldots,\ell$, and let $\ov V:=\ov V_1\otimes \cdots\otimes \ov V_\ell$. Let $\mu\in P^+$ such that $\ov\mu$ is a weight of $\ov V$ and let $k_0:=\sum_{j\ge \hf}\ov\mu(E_j)$. Then $\mu, \la^1,\ldots, \la^\ell\in P^+_{k_0}$.
 \end{lem}

The following corollary is a direct consequence of \propref{tr-KZ}, \thmref{sing-sol} and Lemma \ref{k_0}.
\begin{cor}\label{tr-iso}
Let $n\in \N$ and $\la^i\in P^+$ such that  $\ov\la^i\in\ov P^+_n$, for $i=1,\ldots,\ell$.
Let $\ov V_i$ denote the parabolic Verma module
$\ov\Delta_n(\ov\la^i)$ or irreducible module $\ov L_n(\ov\la^i)$
in $\ov\OO_n$ of highest weight $\ov\la^i$, for $i=1,\ldots,\ell$,
and let $\ov V:=\ov V_1\otimes \cdots\otimes \ov V_\ell$.
Let $\mu\in P^+$ such that $\ov\mu\in\ov P^+_n$ and $\ov\mu$ is a weight of $\ov V$ and let $k_0:=\sum_{j\ge \hf}\ov\mu(E_j)$.
For a fixed $k\ge k_0$, let $V_i$ denote the parabolic Verma module
$\Delta_k(\la^i)$ or irreducible module $L_k(\la^i)$ in $\OO_k$ of highest weight $\la^i$, for $i=1,\ldots,\ell$,
and let $ V:= V_1\otimes \cdots\otimes  V_\ell$.
Then there is a linear isomorphism between the singular solution spaces $\ov{\mc S}(\ov{V}_{\ov\mu})$ and ${\mc S}(V_\mu)$.
\end{cor}

\subsection{KZ equations for modules over $\w{\mathcal {G}}_n$ ${\mathcal {G}}_n$ and $\ov{\mathcal {G}}_n$}\label{nce}
In this subsection, we show that there is a bijection between the sets of solutions
of the (super) KZ equations for tensor product of
modules in $\w\OO_n$ (resp., $\OO_n$ and $\ov\OO_n$)
and for tensor product of
the corresponding $\w{\mathcal {G}}_n$- (resp., ${\mathcal {G}}_n$- and $\ov{\mathcal {G}}_n$-)modules.

For $n\in \N$, let $\mathring{\w{\Omega}}_n$  (resp., $\mathring{\Omega}_n$ and $\mathring{\ov{\Omega}}_n$) denote
 the \emph{Casimir symmetric tensors} of $\w{\mathcal {G}}_n$ (resp., ${\mathcal {G}}_n$ and $\ov{\mathcal {G}}_n$) defined by
 \begin{align}\label{CYTN}\nonumber
 \mathring{\w{\Omega}}_n&:=\sum_{\beta\in\w{\Phi}^+_n}(E^\beta\otimes E_{\beta}+(-1)^{|E_\beta|} E_\beta \otimes E^{\beta})  +\sum_{j\in\w{\I}^+_m(n)}(-1)^{2j}E_j\otimes E_j,\\
 \mathring{\Omega}_n&:=\sum_{\beta\in{\Phi}^+_n} \big(E^\beta\otimes E_{\beta}+(-1)^{|E_\beta|} E_\beta\otimes E^{\beta}) +\sum_{j\in{\I}^+_m(n)}E_j\otimes E_j,\\\nonumber
 \mathring{\ov{\Omega}}_n&:=\sum_{\beta\in\ov{\Phi}^+_n}(E^\beta\otimes E_{\beta}+(-1)^{|E_\beta|} E_\beta \otimes E^{\beta})  +\sum_{j\in\ov{\I}^+_m(n)}(-1)^{2j}E_j\otimes E_j,
  \end{align}
 For $n\in \N$, it is easy to see that
 $\mathring{\w{\Omega}}_n$, $\mathring{{\Omega}}_n$ and  $\mathring{\ov{\Omega}}_n$
  are elements in
  $U(\w{\mathcal {G}}_n)\otimes U(\w{\mathcal {G}}_n)$,
  $U({\mathcal {G}}_n)\otimes U({\mathcal {G}}_n)$
  and
   $U(\ov{\mathcal {G}}_n)\otimes U(\ov{\mathcal {G}}_n)$,
   respectively, satisfying the following equations:
  \begin{align*}
  \mathring{\w{\Omega}}_n&
  =\hf\big(\Delta({\mathring{\w{\bf c}}}_n)- {\mathring{\w{\bf c}}}_n\otimes 1 -1\otimes {\mathring{\w{\bf c}}}_n\big), \\
 \mathring{\Omega}_n&=\hf(\Delta(\mathring{\bf c}_n)-\mathring{\bf c}_n\otimes 1-1\otimes \mathring{\bf c}_n),\\
\mathring{\ov{\Omega}}_n&=\hf(\Delta(\mathring{\ov{\bf c}}_n)-\mathring{\ov{\bf c}}_n\otimes 1-1\otimes \mathring{\ov{\bf c}}_n).
\end{align*}
Recall that $\Delta$ denotes the comultiplication on a universal enveloping algebra. Also it is easy to see
\begin{equation}\label{Om-Om}
 \iota\otimes \iota \big(\mathring{\w\Omega}_n \big)=\w\Omega_n,\ \ \ \
  \iota\otimes \iota \big(\mathring{\Omega}_n\big)=\Omega_n+nK\otimes K\ \ \mbox{and}\ \
  \iota\otimes \iota \big(\mathring{\ov\Omega}_n\big)=\ov\Omega_n-nK\otimes K,
\end{equation}
where $\iota$ are the isomorphisms defined in \eqnref{iso}.

Let $\w{\mc H}_n$ (resp., ${\mc H}_n$ and $\ov{\mc H}_n$) and $\w{\mc B}_n$ (resp., ${\mc B}_n$ and $\ov{\mc B}_n$) denote the Cartan subalgebra  and the Borel subalgebra of $\w{\mc G}_n\oplus \C K$ (resp., ${\mc G}_n\oplus \C K$ and $\ov{\mc G}_n\oplus \C K$) associated to the Dynkin diagram \eqnref{Dynkin:combined}, respectively.
Let $\w{\mc L}_n$ (resp., ${\mc L}_n$ and $\ov{\mc L}_n$) be the standard Levi subalgebra of $\w{\mc G}_n\oplus \C K$ (resp., ${\mc G}_n\oplus \C K$ and $\ov{\mc G}_n\oplus \C K$) associated to $\w Y_n$ (resp., $ Y_n$ and $\ov Y_n$) defined in \eqnref{Y} and let  $\w{\mc P}_n=\w{\mc L}_n+\w{\mc B}_n$ (resp., ${\mc P}_n={\mc L}_n+{\mc B}_n$ and $\ov{\mc P}_n=\ov{\mc L}_n+\ov{\mc B}_n$) be the corresponding parabolic subalgebra.

 For $n\in\N$ and $\mu\in \w{\mc H}_n^*$ (resp., ${\mc H}_n^*$ and $\ov{\mc H}_n^*$), we denote by $\mathring{\w{\Delta}}_n(\mu)=\mbox{Ind}^{\w{\mc G}_n\oplus \C K}_{\w{\mc P}_n} L(\w{\mc L }_n, \mu)$
 (resp., $\mathring{\Delta}_n(\mu)=\mbox{Ind}^{{\mc G}_n\oplus \C K}_{\mc P_n}L(\mc L_n, \mu)$ and
 $\mathring{\ov{\Delta}}_n(\mu)=\mbox{Ind}^{\ov{\mc G}_n\oplus \C K}_{\ov{\mc P}_n} L(\ov{\mc L}_n, \mu)$)
 the parabolic Verma $\w{\mc G}_n\oplus \C K$- (resp., ${\mc G}_n\oplus \C K$- and $\ov{\mc G}_n\oplus \C K$-)module,
 where $ L(\w{\mc L}_n, \mu)$ (resp., $L(\mc L_n, \mu)$ and $ L(\overline{\mc L}_n, \mu)$)
 is the irreducible highest weight $\w{\mc L}_n$- (resp., $\mc L_n$- and $\overline{\mc L}_n$-)module of highest weight
 $\mu$.
 The unique irreducible quotient $\w{\mc G}_n\oplus \C K$- (resp., ${\mc G}_n\oplus \C K$- and $\ov{\mc G}_n\oplus \C K$-)module of
 $\mathring{\w{\Delta}}_n(\w\mu)$ (resp., $\mathring{\Delta}_n(\mu)$ and $\mathring{\ov{\Delta}}_n(\mu)$)
 is denoted by
 $\mathring{\w L}_n(\mu)$ (resp., $\mathring{L}_n(\mu)$ and $\mathring{\overline L}_n(\mu)$).

Since the Cartan subalgebras $\w{\mc H}_n$ (resp., ${\mc H}_n$ and $\ov{\mc H}_n$) and $\w \h_n$ (resp., ${\h}_n$ and $\ov{\h}_n$) are equal, we identify the dual space of  $\w{\mc H}_n$ (resp., ${\mc H}_n$ and $\ov{\mc H}_n$) with the dual space of  $\h_n$ (resp., $\ov{\h}_n$ and $\wt{\h}_n$). Associated to a partition $\la^+=(\la^+_1,\la^+_2,\ldots)$, $d\in\C$ and $\la_{-m},\ldots,\la_{-1}\in\C$ (resp., $\la_{-m},\ldots,\la_{-1}\in\Z$) for $\xx=\mf{a, b, c, d}$ (resp., $\mf b^\bullet$), we define
\begin{align}
\mathring{{\la}} &:=\sum_{i=-m}^{-1}\la_{i}\epsilon_{i}
 + \sum_{j\in\N, 1\le j\le n}(\la^+_{j}-d)\epsilon_{j}
 + d\La_0\in ({\mc H}_n)^{*}, \quad\hbox{ for $\la^+_{n+1}=0$};\label{weight:Im0}\\
\mathring{\ov{\la}} &:=\sum_{i=-m}^{-1}\la_{i}\epsilon_{i}
 + \sum_{s\in\hf+\Z_+,1/2\leq s\leq n-1/2}((\la^+)'_{s+\hf}+d)\epsilon_s
 + d\La_0\in (\ov{\mc H}_n)^*, \quad\hbox{ for $(\la^+)'_{n+1}=0$};\label{weight:ovIm0}\\
\mathring{\w{\la}} &:=\sum_{i=-m}^{-1}\la_{i}\epsilon_{i}
 + \sum_{r\in\hf\N,1/2\leq r\leq n}(\theta(\la^+)_r-(-1)^{2r}d)\epsilon_r
 + d\La_0\in (\wt{\mc H}_n)^*, \quad\hbox{ for $\theta(\la^+)_{1/2+n}=0$}.\label{weight:wtIm0}
\end{align}
Let $\mathring{P}_n^{\raisebox{-10pt}{\scriptsize{+}}}(d)\subset({\mc H}_n)^*$,
$\mathring{\ov{P}}_n^{\raisebox{-10pt}{\scriptsize{+}}}(d)\subset(\ov{\mc H}_n)^*$ and
$\mathring{\wt{P}}_n^{\raisebox{-10pt}{\scriptsize{+}}}(d)\subset (\wt{\mc H}_n)^*$ denote the sets of all weights of the form
\eqnref{weight:Im0}, \eqnref{weight:ovIm0} and \eqnref{weight:wtIm0} for a fixed $d\in \C$,
respectively.
 For $n\in\N$, let
 $\mathring{\wt{P}}_n^{\raisebox{-10pt}{\scriptsize{+}}}:=\bigcup_{d\in \C} \mathring{\wt{P}}_n^{\raisebox{-10pt}{\scriptsize{+}}}(d)$  $ \mathring{{P}}_n^{\raisebox{-10pt}{\scriptsize{+}}}:=\bigcup_{d\in \C}\mathring{{P}}_n^{\raisebox{-10pt}{\scriptsize{+}}}(d)$ and $\mathring{\ov{P}}_n^{\raisebox{-10pt}{\scriptsize{+}}}:=\bigcup_{d\in \C}\mathring{\ov{P}}_n^{\raisebox{-10pt}{\scriptsize{+}}}(d)$.

 For $n\in\N$ and $d\in \C$ ,
 let $\mathring{\w{\mathcal O}}_n(d)$ (resp., $\mathring{\mathcal O}_n(d)$ and $\mathring{\ov{\mathcal O}}_n(d)$)
 be the category of $\w{\mc G}_n\oplus \C K$- (resp., ${\mc G}_n\oplus \C K$- and $\ov{\mc G}_n\oplus \C K$-)modules $M$
 such that $M$ is a semisimple $\w{\mc H}_n$- (resp., $\mc H_n$- and $\overline{\mc H}_n$-)module
 with finite dimensional weight subspaces
 $M_{\gamma}$ for $\gamma\in\w{\mc H}_n^*$ (resp., $\mc H_n^*$ and $\overline{\mc H}_n^*$),
 satisfying
   \begin{itemize}
     \item[ (i)] $M$ decomposes over $\w{\mc L}_n$ (resp., $\mc L_n$ and $\overline{\mc L}_n$)
     as a direct sum of $\w L(\w{\mc L}_n, \mu)$ (resp., $L(\mc L_n, \mu)$ and $\overline L(\overline{\mc L}_n, \mu)$)
     for
     $\mu\in \mathring{\wt{P}}_n^{\raisebox{-10pt}{\scriptsize{+}}}(d)$
      (resp., $\mathring{{P}}_n^{\raisebox{-10pt}{\scriptsize{+}}}(d)$ and $\mathring{\ov{P}}_n^{\raisebox{-10pt}{\scriptsize{+}}}(d)$).
     \item [(ii)] There exist finitely many weights
     $\lambda^1,\ldots,\lambda^k\in \mathring{\wt{P}}_n^{\raisebox{-10pt}{\scriptsize{+}}}(d)$
     (resp., $\mathring{{P}}_n^{\raisebox{-10pt}{\scriptsize{+}}}(d)$
     and
     $\mathring{\ov{P}}_n^{\raisebox{-10pt}{\scriptsize{+}}}(d)$)
      (depending on $M$) such that if $\gamma$ is a weight in $M$,
      then
      $\lambda^i-\gamma\in \sum_{\alpha\in\w\Pi_n}\mathbb{Z}_+\alpha$
       (resp., $\sum_{\alpha\in\Pi_n}\mathbb{Z}_+\alpha$ and $\sum_{\alpha\in\ov\Pi_n}\mathbb{Z}_+\alpha$) for some $i$.
   \end{itemize}
Let $\mathring{\w{\mathcal O}}_n:=\oplus_{d\in \C}\mathring{\w{\mathcal O}}_n(d)$,
$\mathring{{\mathcal O}}_n:=\oplus_{d\in \C}\mathring{\mathcal O}_n(d)$ and
$\mathring{\ov{\mathcal O}}_n:=\oplus_{d\in \C}\mathring{\ov{\mathcal O}}_n(d)$.
The morphisms in the categories are even homomorphisms of modules.
Forgetting the $\Z_2$-gradations,
it is clear that there is an isomorphism $\Psi_n$ of categories from
$\mathring{\w{\mathcal O}}_n(d)$
(resp., $\mathring{\mathcal O}_n(d)$ and $\mathring{\ov{\mathcal O}}_n(d)$)
to ${\w{\mathcal O}}_n(d)$
(resp., ${\mathcal O}_n(d)$ and ${\ov{\mathcal O}}_n(d)$)
induced from the isomorphism $\iota$ defined  in \eqnref{iso} and hence
$\mathring{\w{\mathcal O}}_n$
(resp., $\mathring{\mathcal O}_n$ and $\mathring{\ov{\mathcal O}}_n$) and ${\w{\mathcal O}}_n$
(resp., ${\mathcal O}_n$ and ${\ov{\mathcal O}}_n$)
are isomorphic as tensor categories.
Since  $\Psi_n(M)=M$ for each  $M\in\mathring{\w{\mathcal O}}_n(d)$
 (resp., $\mathring{\mathcal O}_n(d)$ and $\mathring{\ov{\mathcal O}}_n(d)$),
  the $\Z_2$-gradation on $M$ is defined to be the $\Z_2$-gradation on $\Psi_n(M)$.

  For $\mu\in \mathring{\wt{P}}_n^{\raisebox{-10pt}{\scriptsize{+}}}$
   (resp., $\mathring{{P}}_n^{\raisebox{-10pt}{\scriptsize{+}}}$
   and $\mathring{\ov{P}}_n^{\raisebox{-10pt}{\scriptsize{+}}}$),
   the parabolic Verma module $\mathring{\w{\Delta}}_n(\mu)$  (resp., $\mathring{\Delta}_n(\mu)$ and  $\mathring{\ov{\Delta}}_n(\mu)$) and the
 irreducible module   $\mathring{\w L}_n(\mu)$ (resp., $\mathring{L}_n(\mu)$ and $\mathring{\overline L}_n(\mu)$) are in $\mathring{\w\OO}_n$ (resp., $\mathring{\OO}_n$ and $\mathring{\ov\OO}_n$).

Let ${M}:= {M}_1\otimes\cdots\otimes {M}_\ell
 $ for $M_1, \ldots, M_\ell\in \mathring{\w\OO}_n$ (resp., $\mathring{\OO}_n$ and $\mathring{\ov\OO}_n$). Fix a nonzero complex number $\kappa$ and ${\psi}(z_1,\ldots, z_\ell)\in \mc{D}(U, M):=\bigoplus_\mu \mc{D}(U, M_\mu)$, $\mu$ runs over all weights of $M$, we can consider a system of partial differential equations, for $i=1,\ldots,\ell$,
 \begin{equation*}\label{ncKZ}
 \kappa\frac{\partial}{\partial {z_i}}{\psi}=\sum_{j=1\atop j\neq i}^\ell\frac{\mathring{\w{\Omega}}^{(ij)}_n}{z_i-z_j}{\psi},
 \ \ \ \
(\hbox{resp.,}\,\,\, \kappa\frac{\partial}{\partial {z_i}}{\psi}=\sum_{j=1\atop j\neq i}^\ell\frac{\mathring{\Omega}^{(ij)}_n}{z_i-z_j}{\psi},
 \ \ \mbox{and}\ \
 \kappa\frac{\partial}{\partial {z_i}}{\psi}=\sum_{j=1\atop j\neq i}^\ell\frac{\mathring{\ov{\Omega}}^{(ij)}_n}{z_i-z_j}{\psi}\, ).
 \end{equation*}
The systems of equations above are called the \emph{super KZ equations}, \emph{KZ equations} and \emph{super KZ equations}, respectively.
 Let $\mathring{\w{\rm KZ}}(M)$, $\mathring{{\rm KZ}}(M)$ and $\mathring{\ov{\rm KZ}}(M)$ (resp., $\mathring{\w{\mc S}}(M)$, $\mathring{{\mc S}}(M)$ and $\mathring{\ov{\mc S}}(M)$)
 denote the solutions (resp., singular solutions) of the (super) KZ equations in $\mc{D}(U, M)$,
 for ${M}= {M}_1\otimes\cdots\otimes {M}_\ell
 $ and $M_1, \ldots, M_\ell\in \mathring{\w\OO}_n$ (resp., $\mathring{\OO}_n$ and $\mathring{\ov\OO}_n$), respectively.
 The following proposition follows from the direct computation by using \eqnref{Om-Om}.

\begin{prop}\label{ncetoce}
Let $n\in\mathbb{N}$, ${M}= {M}_1\otimes\cdots\otimes {M}_\ell$ and $\psi(z_1,\ldots,z_\ell)\in \mc{D}(U, M)$.
We have
\begin{itemize}
  \item[(i)] $\psi\in\mathring{\w{\rm KZ}}(M)$ (resp., $\mathring{\w{\mc S}}(M)$) if and only if $\psi\in{\w{\rm KZ}}(\Psi_n(M))$ (resp., ${\w{\mc S}}(\Psi_n(M))$),  for ${M}_i\in \mathring{\w\OO}_n(d_i)$, $d_i\in \C$, $i=1,\ldots,\ell$.
  \item[(ii)] $\psi\in\mathring{{\rm KZ}}(M)$ (resp., $\mathring{{\mc S}}(M)$) if and only if
  $\prod^\ell\limits_{1\le i< j\le \ell}(z_i-z_j)^{\frac{-nd_id_j}{\kappa}}\psi\in{{\rm KZ}}(\Psi_n(M))$ (resp., ${{\mc S}}(\Psi_n(M))$),  for ${M}_i\in \mathring{\OO}_n(d_i)$, $d_i\in \C$, $i=1,\ldots,\ell$.
  \item[(iii)] $\psi\in\mathring{\ov{\rm KZ}}(M)$ (resp., $\mathring{\w{\mc S}}(M)$) if and only if
  $\prod^\ell\limits_{1\le i< j\le \ell}(z_i-z_j)^{\frac{nd_id_j}{\kappa}}\psi\in{\ov{\rm KZ}}(\Psi_n(M))$ (resp., ${\ov{\mc S}}(\Psi_n(M))$),  for ${M}_i\in \mathring{\ov\OO}_n(d_i)$, $d_i\in \C$, $i=1,\ldots,\ell$.
\end{itemize}
\end{prop}

\begin{rem}\label{rem:tr-iso}
Let $n\in \N$.
By \corref{tr-iso} and \propref{ncetoce},
we have a bijection between the sets of singular solutions of weight $\ov\la$ of the super KZ equations
for the tensor product of parabolic Verma modules or irreducible modules in $\mathring{\ov\OO}_n$
and the singular solutions of weight $\la$ of the KZ equations
 for the tensor product of the corresponding parabolic Verma modules or irreducible modules in $\mathring{\OO}_k$
 for sufficiently large $k$.

  By the results in \cite[Section 3]{CL}  and \cite[Section 6]{CLW}, the finite dimensional irreducible modules over $\SG_n$ are obtained by applying  the functor $\ov{\mf{tr}}_n$ to some modules in $\ov\OO$. From above, we can obtain the singular solutions of super KZ equations for the tensor product of finite dimensional irreducible modules over $\SG_n$ through the singular solutions of the corresponding KZ equations.
\end{rem}

\section{Trigonometric super KZ equations}\label{Sec:TKZ}
In this section, we study the trigonometric (super) KZ equations associated to the Lie superalgebras $\DG_n$, $\G_n$ and $\SG_n$.
We show that the solutions of the  trigonometric (super) KZ equations are stable under the truncation functors $\w{\mf{tr}}^n_k$, ${\mf{tr}}^n_k$ and $\ov{\mf{tr}}^n_k$. We obtain a bijection between the sets of singular solutions of some special kinds of  trigonometric (super) KZ equations associated to $\DG_\infty$, $\G_\infty$ and $\SG_\infty$. Finally, we have a bijection between the sets of solutions of some special kinds of trigonometric (super) KZ equations associated to $\w{\mc {G}}_n$ (resp., $\mc {G}_n$ and $\ov{\mc {G}}_n$) and associated to  $\DG_n$ (resp., $\G_n$ and $\SG_n$).

Let $n\in\mathbb{N}\cup\{\infty\}$.
Recall  $\w{\Phi}^{+}_n$(resp., ${\Phi}^{+}_n$ and
$\ov{\Phi}^{+}_n$) denotes the set of positive roots of $\DG_n$ (resp., $\G_n$ and $\SG_n$)
 associated to the Dynkin diagrams \eqnref{Dynkin:combined}.
For each $\beta\in\w{\Phi}^{+}_n$(resp., ${\Phi}^{+}_n$ and $\ov{\Phi}^{+}_n$), we fix root vectors $E_\beta$ and $E^\beta$ of weights $\beta$ and $-\beta$, respectively, satisfying $\langle E_{\beta},E^{\beta}\rangle=1.$
By \lemref{finitesum}, $\w{\Omega}_{n,0}$, $\w{\Omega}_{n,+}$, $\w{\Omega}_{n,-}$ (resp., ${\Omega}_{n,0}$, ${\Omega}_{n,+}$, ${\Omega}_{n,-}$, and $\ov{\Omega}_{n,0}$, $\ov{\Omega}_{n,+}$, $\ov{\Omega}_{n,-}$) are operators on  $M\otimes N$, for any $M, N\in \w{\mathcal {O}}_n$ (resp., ${\mathcal {O}}_n$ and $\ov{\mathcal {O}}_n$),
defined by
\begin{eqnarray*}
 &&\w{\Omega}_{n,0}:={\hf} \sum_{j\in\w{\I}^+_m(n)}(-1)^{2j}\big(E_j\otimes E_j-(-1)^{2j}\delta_j  (K\otimes E_j+E_j\otimes K)\big),\\
  &&\w{\Omega}_{n,+}:=\w{\Omega}_{n,0}+\sum_{\beta\in\w{\Phi}^+_n}(-1)^{|E_\beta|} E_\beta\otimes E^{\beta},\ \ \ \
  \w{\Omega}_{n,-}:=\w{\Omega}_{n,0}+\sum_{\beta\in\w{\Phi}^+_n}E^{\beta}\otimes E_\beta,
\end{eqnarray*}
\begin{eqnarray*}
 &&{\Omega}_{n,0}:={\hf} \sum_{j\in{\I}^+_m(n)}\big(E_j\otimes E_j-\delta_j  (K\otimes E_j+E_j\otimes K)\big),\\
  &&{\Omega}_{n,+}:={\Omega}_{n,0}+\sum_{\beta\in{\Phi}^+_n}(-1)^{|E_\beta|} E_\beta\otimes E^{\beta},\ \ \ \
  {\Omega}_{n,-}:={\Omega}_{n,0}+\sum_{\beta\in{\Phi}^+_n}E^{\beta}\otimes E_\beta,
\end{eqnarray*}
\begin{eqnarray*}
 &&\ov{\Omega}_{n,0}:={\hf} \sum_{j\in\ov{\I}^+_m(n)}(-1)^{2j}\big(E_j\otimes E_j+\delta_j  (K\otimes E_j+E_j\otimes K)\big),\\
  &&\ov{\Omega}_{n,+}:=\ov{\Omega}_{n,0}+\sum_{\beta\in\ov{\Phi}^+_n}(-1)^{|E_\beta|} E_\beta\otimes E^{\beta},\ \ \ \
  \ov{\Omega}_{n,-}:=\ov{\Omega}_{n,0}+\sum_{\beta\in\ov{\Phi}^+_n}E^{\beta}\otimes E_\beta.
\end{eqnarray*}
We also define
\begin{equation*}
 \w R_n(z):=\frac{z\w{\Omega}_{n,+}+\w{\Omega}_{n,-}}{z-1},\ \ \ \
 R_n(z):=\frac{z{\Omega}_{n,+}+{\Omega}_{n,-}}{z-1},\ \ \ \
  \ov R_n(z):=\frac{z\ov{\Omega}_{n,+}+\ov{\Omega}_{n,-}}{z-1},\quad\hbox{for $z\in \C\backslash\{1\}$.}
\end{equation*}

 Let ${M}_1, \ldots, M_\ell\in \w\OO_n$  (resp., $\OO_n$ and $\ov\OO_n$)
 for $i=1,\ldots,\ell$ and $ M:= M_1\otimes\cdots\otimes  M_\ell$.
 Let $U$ be a nonempty open subset in the  configuration space ${\bf X}_\ell$.
 Let $\w {\bf h}$ (resp., $ {\bf h}$ and  $\ov {\bf h}$) be an operator defining on every
 $N\in \w\OO_n$  (resp., $\OO_n$ and $\ov\OO_n$) simultaneously such that the restriction of
 $\w {\bf h}$ (resp., $ {\bf h}$ and  $\ov {\bf h}$) to each weight space of $N$ is a multiple of the identity map.
 For a fixed nonzero complex number $\kappa$,  $\w {\bf h}$, $ {\bf h}$,  $\ov {\bf h}$ and ${\psi}(z_1,\ldots, z_\ell)\in \mc{D}(U, M)$, we consider the \emph{trigonometric (super) KZ equations} associated to the Lie (super)algebras $\DG_n$, $\G_n$ and $\SG_n$,  respectively, defined by
\begin{align}
 &(\kappa z_i\frac{\partial}{\partial {z_i}}-\sum_{j=1\atop j\neq i}^\ell\w{R}^{(ij)}_{n}(\frac{z_i}{z_j})-\w {\bf h}^{(i)}){\psi}=0,
 &{\psi}\in \mc{D}(U, M), \quad\mbox{for}\ i=1,\ldots,\ell; \label{dTKZE_n} \\
 & (\kappa  z_i\frac{\partial}{\partial {z_i}}-\sum_{j=1\atop j\neq i}^\ell{R}^{(ij)}_{n}(\frac{z_i}{z_j})-{\bf h}^{(i)}){\psi}=0,
 & {\psi}\in \mc{D}(U, M), \quad \mbox{for}\ i=1,\ldots,\ell; \label{TKZE_n} \\
 &(\kappa  z_i\frac{\partial}{\partial {z_i}}-\sum_{j=1\atop j\neq i}^\ell\ov{R}^{(ij)}_{n}(\frac{z_i}{z_j})-\ov {\bf h}^{(i)}){\psi}=0,
 &{\psi}\in \mc{D}(U, M), \quad \mbox{for}\ i=1,\ldots,\ell,  \label{sTKZE_n}
\end{align}

For a weight $\mu$ of $M$, we define
 \[\w{\rm TKZ}(M_\mu,\w {\bf h}):=\{\psi\in \mc{D}(U, M_\mu)\,|\, \hbox{$\psi$ is a solution of the super KZ equations \eqnref{dTKZE_n}}\}
 \]
and
\[
\w{\rm TKZ}(M,\w {\bf h}):=\bigoplus_\mu\w{\rm TKZ}(M_\mu,\w {\bf h}),
\]
where $\mu$ runs over all weights of $M$.
Analogously, we can define  ${\rm TKZ}(M_\mu,{\bf h})$, $ \ov{\rm TKZ}(M_\mu,\ov {\bf h})$, ${\rm TKZ}(M,{\bf h})$ and $ \ov{\rm TKZ}(M,\ov {\bf h})$
  for the corresponding settings of the trigonometric (super) KZ equations \eqnref{TKZE_n} and \eqnref{sTKZE_n}, respectively.
  The singular solutions of the trigonometric (super) KZ equations are defined
similar to the  singular solutions of the (super) KZ equations.
Let $\w{\mathscr{S}}(M_\mu,\w {\bf h})$ (resp., ${\mathscr{S}}(M_\mu, {\bf h})$ and $\ov{\mathscr{S}}(M_\mu,\ov {\bf h})$) denote the set of functions $\psi\in \w{\rm TKZ}({M}_\mu,\w {\bf h})$ (resp., ${\rm TKZ}(M_\mu,{\bf h})$ and $\ov{\rm TKZ}(M_\mu,\ov {\bf h})$) satisfying $E_\beta \psi=0$ for $\beta\in\w{\Phi}^+$ (resp., ${\Phi}^+$ and $\ov{\Phi}^+$).
The sets are called the \emph{singular solution spaces} of the trigonometric (super) KZ equations and the functions in the sets are called  \emph{singular solutions} of the trigonometric (super) KZ equations.

In contrast to the solutions of the (super) KZ equations, the solutions of
the trigonometric (super) KZ equations may not be stable under the actions of the associated Lie (super)algebras. However, the solutions of trigonometric KZ equations for some special parameters $\bf h$ can be obtained from  the solutions of some KZ equations (see, for example, \cite[ Section 3.8]{EFK}). For these special parameters,
we can obtain the analogous results of (super) KZ equations given in Section~\ref{RKZ} for the trigonometric (super) KZ equations with some tedious computations.

Applying a similar proof of  \propref{tr-KZ}, we have the following proposition.
\begin{prop}\label{tr-TKZ} Let $0\le k<n \leq\infty$ and $M=M_1\otimes\cdots\otimes  M_\ell$. Assume that  ${\psi}\in \mc{D}(U, M)$.
\begin{itemize}
\item[(i)] For $\mu\in \w\Xi_k$ and $M_1, \ldots,M_\ell\in \w{\mc {O}}_n$, we have
\[
{\psi}\in \w{\rm TKZ}(M_\mu,\w {\bf h}) \Leftrightarrow \w{\mf{tr}}^n_k({\psi})\in \w{\rm TKZ}\big(\w{\mf{tr}}^n_k(M)_\mu,\w {\bf h}\big)
\,\,\hbox{and}\,\,
{\psi}\in \w{{\mathscr{S}}}(M_\mu, \w {\bf h}) \Leftrightarrow \w{\mf{tr}}^n_k({\psi})\in \w{{\mathscr{S}}}\big(\w{\mf{tr}}^n_k(M)_\mu, \w {\bf h}\big).
\]
\item[(ii)] For $\mu\in \Xi_k$ and $M_1, \ldots,M_\ell\in{\mc {O}}_n$, we have
 \[
{\psi}\in {\rm TKZ}(M_\mu, {\bf h}) \Leftrightarrow {\mf{tr}}^n_k({\psi})\in {\rm TKZ}\big({\mf{tr}}^n_k(M)_\mu, {\bf h}\big)
\,\,\hbox{and}\,\,
{\psi}\in {{\mathscr{S}}}(M_\mu,  {\bf h}) \Leftrightarrow {\mf{tr}}^n_k({\psi})\in {{\mathscr{S}}}\big({\mf{tr}}^n_k(M)_\mu,  {\bf h}\big).
\]
\item[(iii)] For $\mu\in \ov\Xi_k$ and $M_1, \ldots,M_\ell\in\ov{\mc {O}}_n$, we have
\[
{\psi}\in \ov{\rm TKZ}(M_\mu,\ov {\bf h}) \Leftrightarrow \ov{\mf{tr}}^n_k({\psi})\in \ov{\rm TKZ}\big(\ov{\mf{tr}}^n_k(M)_\mu,\ov {\bf h}\big)
\,\,\hbox{and}\,\,
{\psi}\in \ov{{\mathscr{S}}}(M_\mu, \ov {\bf h}) \Leftrightarrow \ov{\mf{tr}}^n_k({\psi})\in \ov{{\mathscr{S}}}\big(\ov{\mf{tr}}^n_k(M)_\mu, \ov {\bf h}\big).
\]
\end{itemize}
\end{prop}

For $\gamma\in\w{\mathfrak{h}}^*_n$ (resp., ${\mathfrak{h}}^*_n$ and $\ov{\mathfrak{h}}^*_n$), we define the operator $\w {\bf h}_\gamma$ (resp., $ {\bf h}_\gamma$  and $\ov {\bf h}_\gamma$) on $N\in \w\OO_n$  (resp., $\OO_n$ and $\ov\OO_n$) by $\w {\bf h}_\gamma (v)=(\gamma,\mu)v$ (resp., $ {\bf h}_\gamma (v)=(\gamma,\mu)v$ and $\ov {\bf h}_\gamma (v)=(\gamma,\mu)v$) for any weight vector $v\in N$ of the weight $\mu$.

Then we have that
\begin{lem}\label{tec lem}
Let $n \in \N\cup\{\infty\}$ and let $v$ be a weight vector of  the weight $\mu$ in ${M}:= {M}_1\otimes\cdots\otimes {M}_\ell$.
\begin{itemize}
\item[(i)] For $\mu\in \w P^+_n$ and $M_1, \ldots,M_\ell\in \w{\mc {O}}_n$, we have
 \[
 \big(\sum_{j=1\atop j\neq i}^\ell{\w\Omega}^{(ij)}_{n, -}-(\hf \w{\bf h}_\mu+\w\varrho_n)^{(i)}\big)(v)= \big(\sum_{\beta\in{\w\Phi_n}^{+}}E^{\beta(i)}E_{\beta}-\hf \w{\bf c}_n^{(i)}\big) (v).
 \]
\item[(ii)] For $\mu\in  P^+_n$ and $M_1, \ldots,M_\ell\in{\mc {O}}_n$, we have
 \[
 \big(\sum_{j=1\atop j\neq i}^\ell{\Omega}^{(ij)}_{n, -}-(\hf {\bf h}_\mu+\varrho_n)^{(i)}\big)(v)= \big(\sum_{\beta\in{\Phi_n^+}}E^{\beta(i)}E_{\beta}-\hf {\bf c}_n^{(i)}\big) (v).
 \]
\item[(iii)] For $\mu\in \ov P^+_n$ and $M_1,\ldots,M_\ell\in\ov{\mc {O}}_n$, we have
 \[
 \big(\sum_{j=1\atop j\neq i}^\ell{\ov\Omega}^{(ij)}_{n, -}-(\hf \ov{\bf h}_\mu+\ov\varrho_n)^{(i)}\big)(v)= \big(\sum_{\beta\in{\ov\Phi}_n^{+}}E^{\beta(i)}E_{\beta}-\hf \ov{\bf c}_n^{(i)}\big) (v).
 \]
\end{itemize}
\end{lem}
\begin{proof}
We will prove (i). The proof of other cases are similar. We may assume that $v=v_1\otimes\cdots \otimes v_\ell$ such that $v_i$ is a weight vector of weight $\gamma^i$ in $M_i$, for each $i=1,\ldots,\ell$.
\begin{eqnarray*}
 && \big(\sum_{j=1\atop j\neq i}^\ell{\w\Omega}^{(ij)}_{n, -}-(\hf \w{\bf h}_\mu+\w\varrho_n)^{(i)}\big)(v)\\
 &= &\big(\sum_{j=1\atop j\neq i}^\ell
 \big(\frac{1}{2}\sum_{k\in\w\I^+_m(n)}((-1)^{2k}E^{(i)}_kE^{(j)}_k-\delta_k(K^{(i)}E^{(j)}_k+E^{(i)}_kK^{(j)}))
 +\sum_{\beta\in{\w\Phi_n}^{+}}E^{\beta(i)}E^{(j)}_{\beta}\big)
-(\hf \w{\bf h}_\mu+\w\varrho_n)^{(i)}\big)(v)\\
 &=&
  \big(\frac{1}{2}\sum_{k\in\w\I^+_m(n)}((-1)^{2k}E^{(i)}_kE_k-\delta_k(K^{(i)}E_k+E^{(i)}_kK))
 +\sum_{\beta\in{\w\Phi_n}^{+}}E^{\beta(i)}E_{\beta}\big) (v)\\
 &&\,\,
 -\big(\frac{1}{2}\sum_{k\in\w\I^+_m(n)}((-1)^{2k}E^{(i)}_kE^{(i)}_k-\delta_k(K^{(i)}E^{(i)}_k+E^{(i)}_kK^{(i)}))
 +\sum_{\beta\in{\w\Phi_n}^{+}}E^{\beta(i)}E^{(i)}_{\beta}+\w\varrho_n^{(i)}+\hf \w{\bf h}_\mu^{(i)}\big) (v)\\
  &=&  \big(\frac{1}{2}(\mu,\gamma^i)+\sum_{\beta\in{\w\Phi_n}^{+}}E^{\beta(i)}E_{\beta}\big) (v)-(\hf \w{\bf c}_n^{(i)}+\hf \w{\bf  h}_\mu^{(i)}\big) (v)\\
  &=&  \big(\sum_{\beta\in{\w\Phi_n}^{+}}E^{\beta(i)}E_{\beta}-\hf \w{\bf c}_n^{(i)}\big) (v).
   \end{eqnarray*}
  \end{proof}

 The following theorem is an analogy of \thmref{sing-sol}. Recall that $X_\la$, $\ov X_\la$, $Y_\la$ and $\ov Y_\la$ are defined in \lemref{matching:weights}.

 \begin{thm}
\label{sing-sol-t}
Let $\la^i\in P^+$ and let $\w V(\w{\la}^i)\in \w{\mathcal O}$ be a highest weight module of weight $\w{\la}^i$,  for $i=1,\ldots,\ell$.
Let $\w M:=\w V(\w{\la}^1)\otimes\cdots\otimes \w V(\w{\la}^\ell)$, $M:=T(\w M)=T(\w V(\w{\la}^1))\otimes\cdots\otimes T(\w V(\w{\la}^\ell))$ and $\ov M:=\ov T(\w M)=\ov T(\w V(\w{\la}^1))\otimes\cdots\otimes \ov T(\w V(\w{\la}^\ell))$.
 For $\la\in  P^+$, there are linear isomorphisms
\begin{equation*}
\begin{aligned}[c]
\varphi : \, \w{{\mathscr{S}}}(\w{M}_{\w\la}, \w {\bf h}_{\w \la}+\w\varrho) & \longrightarrow {{\mathscr{S}}}(M_{\la},  {\bf h}_{\la} +\varrho)\\
\psi & \mapsto \emph{Y}_\la \psi
\end{aligned}
\quad \hbox{and}\quad
\begin{aligned}[c]
\phi :\, \w{\mathscr{S}}(\w{M}_{\w\la}, \w {\bf h}_{\w \la}+\w\varrho)& \longrightarrow \ov{{\mathscr{S}}}({\ov M}_{\ov\la}, \ov {\bf h}_{\ov \la} +\ov\varrho)\\
 \psi & \mapsto \ov{\emph{Y}}_\la \psi
\end{aligned}
\end{equation*}
with $\varphi^{-1}(\psi )=\emph{X}_\la \psi $,  for $\psi\in{\mathscr{S}}({M}_{\la}, {\bf h}_{ \la} +\varrho)$ and $\phi^{-1}(\ov \psi )=\ov{\emph{X}}_\la \ov \psi $, for $\ov \psi\in{\mathscr{S}}(\ov{M}_{\ov\la},\ov {\bf h}_{\ov \la} +\ov\varrho)$.
\end{thm}
\begin{proof}
We will prove the case for $\varphi$ is an isomorphism  and $\varphi^{-1}(\psi )=\emph{X}_\la \psi $  for $\psi\in{\mathscr{S}}({M}_{\la}, {\bf h}_{ \la} +\varrho)$. The other case is similar.
Fix $i\in \{1,\ldots,\ell\}$. Let $ \psi\in \mc{D}(U, \w M_{\w \la})$ such that $ \psi(z_1,\ldots, z_\ell)$ is a singular vector in $\w M$ for each $(z_1,\ldots, z_\ell)\in U$. By \lemref{tec lem} and \lemref{slem:Gamma}, we have
\begin{eqnarray*}
\big(\sum\limits_{j=1\atop j\neq i}^\ell\w{R}^{(ij)}(\frac{z_i}{z_j})+({\hf}\w{\bf h}_{\w\la}+\w\varrho)^{(i)}\big)(\psi)
&=&z_i\sum_{j=1\atop j\neq i}^\ell{\frac{\Omega^{(ij)}}{z_i-z_j}}(\psi)-
\big(\sum_{j=1\atop j\neq i}^\ell{\w\Omega}^{(ij)}_{-}-({\hf}\w{\bf h}_{\w\la}+\w\varrho)^{(i)}\big)(\psi)\\\nonumber
&=&z_i\sum_{j=1\atop j\neq i}^\ell\frac{{\w\Omega}^{(ij)}}{z_i-z_j}(\psi)-
\big(\sum_{\beta\in{\w\Phi_n}^{+}}E^{\beta(i)}E_{\beta}-\hf \w{\bf c}_n^{(i)}\big)(\psi)\\\nonumber
&=&z_i\sum_{j=1\atop j\neq i}^\ell\frac{{\w\Omega}^{(ij)}}{z_i-z_j}(\psi)-\hf(\w{\la}^i+2\w\rho,\w{\la}^i)\psi.
\end{eqnarray*}
Briefly, we have
\begin{equation}\label{t1}
\big(\sum\limits_{j=1\atop j\neq i}^\ell\w{R}^{(ij)}(\frac{z_i}{z_j})+({\hf}\w{\bf h}_{\w\la}+\w\varrho)^{(i)}\big)(\psi)
 = z_i\sum_{j=1\atop j\neq i}^\ell\frac{{\w\Omega}^{(ij)}}{z_i-z_j}(\psi)-\hf(\w{\la}^i+2\w\rho,\w{\la}^i)\psi.
\end{equation}
Similarly,
\begin{equation}\label{t2}
 \big(\sum\limits_{j=1\atop j\neq i}^\ell{R}^{(ij)}(\frac{z_i}{z_j})+({\hf}{\bf h}_{\la}+\varrho)^{(i)}\big)(\psi)
 =z_i\sum_{j=1\atop j\neq i}^\ell\frac{{\Omega}^{(ij)}}{{z_i-z_j}}(\psi)-\hf({\la^i}+2\rho,{\la^i})\psi,
\end{equation}
for $i\in \{1,\ldots,\ell\}$, where $ \psi\in \mc{D}(U, T(\w M)_{\w \la})$ such that $\psi(z_1,\ldots,z_\ell)$ is a singular vector in $T(\w M)$ for each $(z_1,\ldots,z_\ell)\in U$.

Now, let $ \psi\in \w{{\mathscr{S}}}(\w{M}_{\w\la}, \w {\bf h}_{\w \la}+\w\varrho)$. By  \eqnref{t1}, \propref{thm:scalar}, \propref{Om-g} and \lemref{wOm-Om}, we have
\begin{eqnarray*}
Y_\la\big(\sum\limits_{j=1\atop j\neq i}^\ell\w{R}^{(ij)}(\frac{z_i}{z_j})+({\hf}\w{\bf h}_{\w\la}+\w\varrho)^{(i)}\big)(\psi)
= z_i\sum_{j=1\atop j\neq i}^\ell\frac{{\w\Omega}^{(ij)}}{{z_i-z_j}}(Y_\la\psi)-\hf({\w\la^i}+2\w\rho,{\w\la^i})Y_\la\psi\\
 = z_i\sum_{j=1\atop j\neq i}^\ell\frac{{\Omega}^{(ij)}}{{z_i-z_j}}(Y_\la\psi)-\hf({\la^i}+2\rho,{\la^i})Y_\la\psi.
\end{eqnarray*}
By \eqnref{t2}, we have $Y_\la\psi\in {{\mathscr{S}}}(M_{\la},  {\bf h}_{\la} +\varrho)$ since $ \psi\in \w{{\mathscr{S}}}(\w{M}_{\w\la}, \w {\bf h}_{\w \la}+\w\varrho)$.
On the other hand, we assume $\psi\in {{\mathscr{S}}}( M_{\la},  {\bf h}_{\la} +\varrho)$.
By \eqnref{t1}, \lemref{wOm-Om}, \propref{Om-g} and \propref{thm:scalar}, we have
\begin{eqnarray*}
X_\la\big(\sum\limits_{j=1\atop j\neq i}^\ell{R}^{(ij)}(\frac{z_i}{z_j})+({\hf}{\bf h}_{\la}+\varrho)^{(i)}\big)(\psi)
= X_\la z_i\sum_{j=1\atop j\neq i}^\ell\frac{\w{\Omega}^{(ij)}}{{z_i-z_j}}(\psi)-\hf({\la^i}+2\rho,{\la^i})X_\la\psi\\
= z_i\sum_{j=1\atop j\neq i}^\ell\frac{{\w\Omega}^{(ij)}}{{z_i-z_j}}(X_\la\psi)-\hf(\w{\la}^i+2\w\rho,\w{\la}^i)X_\la\psi.
 \end{eqnarray*}
By \eqnref{t1}, we have $X_\la\psi\in \w{{\mathscr{S}}}(\w{M}_{\w\la}, \w {\bf h}_{\w \la}+\w\varrho)$  since $\psi\in {{\mathscr{S}}}(M_{\la},  {\bf h}_{\la} +\varrho)$.
Now, the proof is completed by \lemref{matching:weights}.
\end{proof}

The following corollary is analogous to  \corref{tr-iso} and it is  a direct consequence of \propref{tr-TKZ} and  \thmref{sing-sol-t}.

\begin{cor}\label{tr-iso-t}
Let $n\in \N$ and $\la^i\in P^+$ such that  $\ov\la^i\in\ov P^+_n$, for $i=1,\ldots,\ell$.
Let $\ov V_i$ denote the parabolic Verma module $\ov\Delta_n(\ov\la^i)$ or irreducible module $\ov L_n(\ov\la^i)$ in $\ov\OO_n$ of highest weight $\ov\la^i$, for $i=1,\ldots,\ell$, and let $\ov V:=\ov V_1\otimes \cdots\otimes \ov V_\ell$. Let $\mu\in P^+$ such that $\ov\mu\in\ov P^+_n$ and $\ov\mu$ is a weight of $\ov V$ and let $k_0:=\sum_{j\ge \hf}\ov\mu(E_j)$.
For a fixed $k\ge k_0$, let $V_i$ denote the parabolic Verma module $\Delta_k(\la^i)$ or irreducible module $L_k(\la^i)$ in $\OO_k$ of highest weight $\la^i$, for $i=1,\ldots,\ell$, and let $ V:= V_1\otimes \cdots\otimes  V_\ell$. Then there is a linear isomorphism between the singular solution spaces
$\ov{{\mathscr{S}}}(\ov V_{\ov\mu}, \ov {\bf h}_{\ov \mu} +\ov\varrho_n)$ and ${{\mathscr{S}}}(V_{\mu},  {\bf h}_{\mu} +\varrho_k)$.
\end{cor}

\subsection{Trigonometric (super) KZ equations for modules over $\w{\mathcal {G}}_n$ ${\mathcal {G}}_n$ and $\ov{\mathcal {G}}_n$}\label{nce-t}
In this subsection, we show that there is a bijection between the sets of solutions
of the trigonometric (super) KZ equations for tensor product of modules in $\w\OO_n$ (resp., $\OO_n$ and $\ov\OO_n$)
and for tensor product of
the corresponding $\w{\mathcal {G}}_n$- (resp., ${\mathcal {G}}_n$- and $\ov{\mathcal {G}}_n$-)modules.

 Recall that  the Casimir symmetric tensors $\mathring{\w{\Omega}}_n$  (resp., $\mathring{\Omega}_n$ and $\mathring{\ov{\Omega}}_n$)  for $\mathring{\w{\mathcal {G}}}_n$ (resp., $\mathring{{\mathcal {G}}}_n$ and $\mathring{\ov{\mathcal {G}}}_n$) are defined
 in \eqref{CYTN}. For $n\in \N$, we define the operators $\mathring{\w{\Omega}}_{n,0}$, $\mathring{\w{\Omega}}_{n,+}$, $\mathring{\w{\Omega}}_{n,-}$ (resp., $\mathring{\Omega}_{n,0}$, $\mathring{\Omega}_{n,+}$, $\mathring{\Omega}_{n,-}$, and $\mathring{\ov{\Omega}}_{n,0}$, $\mathring{\ov{\Omega}}_{n,+}$, $\mathring{\ov{\Omega}}_{n,-}$)  on  $M\otimes N$, for any $M, N\in \w{\mathcal {O}}_n$ (resp., ${\mathcal {O}}_n$ and $\ov{\mathcal {O}}_n$), by
\begin{eqnarray*}
 &\mathring{\w{\Omega}}_{n,0}:={\hf} \sum_{j\in\w{\I}^+_m(n)}(-1)^{2j}E_j\otimes E_j,\\
  &\mathring{\w{\Omega}}_{n,+}:=\mathring{\w{\Omega}}_{n,0}+\sum_{\beta\in\w{\Phi}^+_n}(-1)^{|E_\beta|} E_\beta\otimes E^{\beta},\qquad
  \mathring{\w{\Omega}}_{n,-}:=\w{\Omega}_{n,0}+\sum_{\beta\in\w{\Phi}^+_n}E^{\beta}\otimes E_\beta,
\end{eqnarray*}
\begin{eqnarray*}
 &\mathring{\Omega}_{n,0}:={\hf} \sum_{j\in{\I}^+_m(n)}E_j\otimes E_j,\\ &\mathring{\Omega}_{n,+}:=\mathring{\Omega}_{n,0}+\sum_{\beta\in{\Phi}^+_n}(-1)^{|E_\beta|} E_\beta\otimes E^{\beta},\qquad
  \mathring{\Omega}_{n,-}:=\mathring{\Omega}_{n,0}+\sum_{\beta\in{\Phi}^+_n}E^{\beta}\otimes E_\beta,
\end{eqnarray*}
\begin{eqnarray*}
& \mathring{\ov{\Omega}}_{n,0}:={\hf} \sum_{j\in\ov{\I}^+_m(n)}(-1)^{2j}E_j\otimes E_j,\\ &\mathring{\ov{\Omega}}_{n,+}:=\mathring{\ov{\Omega}}_{n,0}+\sum_{\beta\in\ov{\Phi}^+_n}(-1)^{|E_\beta|} E_\beta\otimes E^{\beta},\qquad
  \mathring{\ov{\Omega}}_{n,-}:=\mathring{\ov{\Omega}}_{n,0}+\sum_{\beta\in\ov{\Phi}^+_n}E^{\beta}\otimes E_\beta.
\end{eqnarray*}
We also define
\begin{equation*}
 \mathring{\w R}_n(z):=\frac{z\mathring{\w{\Omega}}_{n,+}+\mathring{\w{\Omega}}_{n,-}}{z-1},\ \ \ \
 \mathring{R}_n(z):=\frac{z\mathring{\Omega}_{n,+}+\mathring{\Omega}_{n,-}}{z-1},\ \ \ \
  \mathring{\ov R}_n(z):=\frac{z\mathring{\ov{\Omega}}_{n,+}+\mathring{\ov{\Omega}}_{n,-}}{z-1},\quad\hbox{for $z\in \C\backslash\{1\}$.}
\end{equation*}

Recall that we identify the Cartan subalgebras
$\w{\mc H}_n$ (resp., ${\mc H}_n$ and $\ov{\mc H}_n$)
with
$\w \h_n$ (resp., ${\h}_n$ and ${\h}_n$)
and we identify the dual space of
$\w{\mc H}_n$ (resp., ${\mc H}_n$ and $\ov{\mc H}_n$)
with the dual space of
$\w\h_n$ (resp., ${\h}_n$ and $\ov{\h}_n$).
Let $\w {\bf h}$ (resp., $ {\bf h}$ and  $\ov {\bf h}$)
be an operator defining on every
$N\in \w\OO_n$  (resp., $\OO_n$ and $\ov\OO_n$)
simultaneously such that the restriction of
$\w {\bf h}$ (resp., $ {\bf h}$ and  $\ov {\bf h}$)
to each weight space of $N$ is a multiple of the identity map.
Let $U$ be a nonempty open subset in the  configuration space ${\bf X}_\ell$.
Let ${M}:= {M}_1\otimes\cdots\otimes {M}_\ell
 $ for $M_1, \ldots, M_\ell\in \mathring{\w\OO}_n$ (resp., $\mathring{\OO}_n$ and $\mathring{\ov\OO}_n$). We refer to Subsection \ref{nce} for the definition of  $\mathring{\w\OO}_n$  (resp., $\mathring{\OO}_n$ and $\mathring{\ov\OO}_n$).
 For a fixed nonzero complex number $\kappa$,  $\w {\bf h}$, $ {\bf h}$,  $\ov {\bf h}$ and ${\psi}(z_1,\ldots, z_\ell)\in \mc{D}(U, M)$, we consider the \emph{trigonometric (super) KZ equations} associated to
  $\w{\mathcal{G}}_n$, $\mathcal{G}_n$ and $\ov{\mathcal{G}}_n$, respectively,
 defined by
\begin{equation}\label{ndTKZE_n}
(\kappa z_i\frac{\partial}{\partial {z_i}}-\sum_{j=1\atop j\neq i}^\ell\mathring{\w{R}}^{(ij)}_{n}(\frac{z_i}{z_j})-\w{\bf h}^{(i)}){\psi}=0,  \quad {\psi}\in \mc{D}(U, M), \ \ \ \ \mbox{for}\ i=1,\ldots,\ell;
\end{equation}
\begin{equation}\label{nTKZE_n}
(\kappa z_i\frac{\partial}{\partial {z_i}}-\sum_{j=1\atop j\neq i}^\ell\mathring{R}^{(ij)}_{n}(\frac{z_i}{z_j})-{\bf h}^{(i)}){\psi}=0,\quad {\psi}\in \mc{D}(U, M), \ \ \ \ \mbox{for}\ i=1,\ldots,\ell;
\end{equation}
\begin{equation}\label{nsTKZE_n}
(\kappa z_i\frac{\partial}{\partial {z_i}}-\sum_{j=1\atop j\neq i}^\ell\mathring{\ov{R}}^{(ij)}_{n}(\frac{z_i}{z_j})-\ov{\bf h}^{(i)}){\psi}=0, \quad {\psi}\in \mc{D}(U, M),\ \ \ \ \mbox{for}\ i=1,\ldots,\ell.
\end{equation}

Let $\mathring{\w{\rm TKZ}}(M_\mu,\mbox{$\w{\bf h}$})$ (resp., $\mathring{{\rm TKZ}}(M_\mu,{\bf  h})$ and $\mathring{\ov{\rm TKZ}}(M_\mu,\ov{\bf  h})$)
denote the set of the solutions $\psi\in \mc{D}(U, M_\mu)$ for the trigonometric (super) KZ equations
\eqref{ndTKZE_n} (resp., \eqref{nTKZE_n} and \eqref{nsTKZE_n}) for each weight $\mu$ of $M$. Let $\mathring{\w{{\mathscr{S}}}}(M_\mu,\w{\bf  h})$ (resp., $\mathring{{{\mathscr{S}}}}(M_\mu,{\bf  h})$ and $\mathring{\ov{{\mathscr{S}}}}(M_\mu,\ov{\bf  h})$) be a subset of $\mathring{\w{\rm TKZ}}(M_\mu,\mbox{$\w{\bf h}$})$ (resp., $\mathring{{\rm TKZ}}(M_\mu,{\bf  h})$ and $\mathring{\ov{\rm TKZ}}(M_\mu,\ov{\bf  h})$) consisting of singular solutions.

Note that we have
\begin{eqnarray}
 & \iota\otimes \iota \big( {\mathring{\w R}}{}^{(ij)}_{n}(\frac{z_i}{z_j})\big)
={\w R}^{(ij)}_{n}(\frac{z_i}{z_j}), &\hbox{for $i\not= j$};\nonumber\\
 & \iota\otimes \iota \big({\mathring{R}}^{(ij)}_{n}(\frac{z_i}{z_j})\big) =R^{(ij)}_{n}(\frac{z_i}{z_j})+\frac{n(z_i+z_j)}{2(z_i-z_j)} K\otimes K,  &\hbox{for $i\not= j$};\label{Om_+O_+}\\
 &  \iota\otimes \iota \big( \mathring{\ov R}{}^{(ij)}_{n}(\frac{z_i}{z_j})\big)
  = \mathring{\ov R}{}^{(ij)}_n(\frac{z_i}{z_j})-\frac{n(z_i+z_j)}{2(z_i- z_j)} K\otimes K, & \hbox{for $i\not= j$}.\nonumber
\end{eqnarray}

Analogous to \propref{ncetoce}, we have the following proposition using \eqnref{Om_+O_+} by direct computation.

\begin{prop}\label{Tncetoce}
Let $n\in\mathbb{N}$, ${M}= {M}_1\otimes\cdots\otimes {M}_\ell$ and $\psi(z_1,\ldots,z_\ell)\in \mc{D}(U, M)$.
\begin{itemize}
  \item[(i).] $\psi\in\mathring{\w{\rm TKZ}} (M, \mbox{$\w{\bf h}$})$ (resp., $\mathring{\w{{\mathscr{S}}}} (M, \mbox{$\w{\bf h}$})$) if and only if
  $\psi\in\w{\rm TKZ}(\Psi_n(M),{\w {\bf h}})$ (resp., ${\w{{\mathscr{S}}}} (\Psi_n(M),{\w {\bf h}})$), for ${M}_i\in \mathring{\w\OO}_n(d_i)$, $d_i\in \C$, $i=1,\ldots,\ell$.
  \item[(ii).] $\psi\in\mathring{{\rm TKZ}}(M,{\bf h})$ (resp., $\mathring{{{\mathscr{S}}}} (M, \mbox{${\bf h}$})$) if and only if
  $\prod^\ell\limits_{1\le i< j\le \ell}(z_i-z_j)^{\frac{-nd_id_j}{\kappa}}\prod\limits_{r=1}^\ell z_r^\frac{nd_r(d-d_r)}{2\kappa}\psi\in{{\rm TKZ}}(\Psi_n(M),{\bf h})$ (resp., ${{{\mathscr{S}}}} (\Psi_n(M),{ {\bf h}})$), for ${M}_i\in \mathring{\OO}_n(d_i)$, $d_i\in \C$, $i=1,\ldots,\ell$ and $d=\sum_{i=1}^\ell d_i$.
  \item[(iii).] $\psi\in\mathring{\ov{\rm TKZ}}(M,\ov{\bf h})$ (resp., $\mathring{\ov{{\mathscr{S}}}} (M, \mbox{$\ov{\bf h}$})$) if and only if
  $\prod^\ell\limits_{1\le i< j\le \ell}(z_i-z_j)^{\frac{nd_id_j}{\kappa}}\prod\limits_{r=1}^\ell z_r^{\frac{-nd_r(d-d_r)}{2\kappa}}\psi\in{\ov{\rm TKZ}}(\Psi_n(M),\ov{\bf h})$ (resp., ${\ov{{\mathscr{S}}}} (\Psi_n(M),{\ov {\bf h}})$), for ${M}_i\in \mathring{\ov\OO}_n(d_i)$,  $d_i\in \C$, $i=1,\ldots,\ell$
  and $d=\sum_{i=1}^\ell d_i$.
\end{itemize}
\end{prop}

\begin{rem}
We have an analogous result of \remref{rem:tr-iso} for super KZ equations.
By \corref{tr-iso-t} and \propref{ncetoce}, we have a bijection between the sets of singular solutions of weight $\ov\la$ of the trigonometric super KZ equations for the tensor product of parabolic Verma modules or irreducible modules in $\mathring{\ov\OO}_n$ for $\ov{\bf h}=\ov {\bf h}_{\ov \la} +\ov\varrho$ and $n\in \N$ and the singular solutions of weight $\la$ of the trigonometric KZ equations for the tensor product of the corresponding parabolic Verma modules or irreducible modules in $\mathring{\OO}_k$ for ${\bf h}= {\bf h}_{ \la} +\varrho$ and sufficiently large $k$.
\end{rem}

\end{document}